\documentclass[a4paper, 11pt]{article} 
\usepackage{ourpreamble}

\begin{document}

\makeatletter
\if@todonotes@disabled\else
    \tableofcontents\newpage
\fi
\makeatother

\maketitle

\begin{abstract}
    \noindent The one-way model of measurement-based quantum computing implements computations via successive adaptive single-qubit measurements on a resource graph state.
    This model has practical applications, particularly in photonics, and it is also useful as a theoretical tool e.g.\ for optimisation.
    Gflow is a necessary and sufficient condition for implementing certain one-way computations deterministically (in a suitable sense); it is also used in efficient translations from the one-way model to quantum circuits.
    For a computation on $n$ qubits, a gflow can be found in $\mathcal{O}(n^3)$ time.

    Here, we consider an incompletely specified computation given by an unlabelled open graph: the graph state as well as the input and output qubits are known, but the measurements have not yet been fixed.
    We give an algorithm that identifies a measurement labelling and a compatible gflow, and runs in $\mathcal{O}(n^3)$, strictly generalising the previous approach.
    The new algorithm can also handle restrictions on the order of the measurements and returns only solutions compatible with these constraints.
    We additionally prove that if an open graph has equal numbers of inputs and outputs, it has at most one labelling compatible with gflow; and show how to identify additional inputs for an open graph that does not yet have the maximal number, without breaking an existing gflow.
\end{abstract}

\section{Introduction}

Quantum computing offers not only polynomial-time algorithms for some problems where no efficient classical algorithm is known, it also gives rise to entirely non-classical models of computation.
Where quantum circuits are a generalisation of reversible classical logic circuits, the one-way model of measurement-based quantum computing (MBQC) \cite{raussendorfOneWayQuantumComputer2001} is intrinsically quantum.
In this model, an entangled resource state called a graph state is prepared at the beginning, then the computation is driven by successive adaptive single-qubit measurements, which remove from the computation the qubit being measured and simultaneously change the state on the remaining qubits.
As measurements are probabilistic, care is needed to ensure that a computation is deterministic overall.
One option is to choose a universal family of graph states such as the cluster states, for which it is well-understood how to implement arbitrary computations deterministically \cite{raussendorfOneWayQuantumComputer2001}.
Yet using cluster states is somewhat inefficient: this approach generally involves many qubits and measurements which do not meaningfully contribute to the computation.

To improve efficiency, it is better to work with graph states tailored to the specific computation, yet the trade-off is that more work needs to be done to ensure determinism \cite{danosDeterminismOnewayModel2006}.
A graph state is described by a finite simple graph, with graph vertices corresponding to qubits and edges corresponding to entanglement.
An open graph is a graph with two distinguished sets of vertices called `inputs' and `outputs', these correspond to the inputs and outputs of the one-way computation.
Finally, each non-output qubit has a measurement label that fixes one of several families of measurements, for an overall structure called a \emph{labelled open graph}.
A procedure for ensuring determinism is called a \emph{flow}, it consists of a partial order on the qubits (indicating which measurements need to be happen earlier or later), and a way of correcting each undesired measurement outcome \cite{danosDeterminismOnewayModel2006,browneGeneralizedFlowDeterminism2007}.
In this work, we will focus on the flow variant called gflow, which is necessary and sufficient for robust determinism when measurements belong to three families\footnote{Some earlier literature uses `gflow' to refer only to the flow property on \LOG{}s where all measurements are $\XY$ and uses `extended gflow' to denote the general case. We will instead specify `$\XY$-only gflow' where needed.} denoted $\XY,\XZ,\YZ$ \cite{browneGeneralizedFlowDeterminism2007}.
This choice allows universal quantum computation.
Gflow is one example in an entire family of flow properties that differ in terms of which measurement types and which kinds of corrections are allowed \cite{danosDeterminismOnewayModel2006,browneGeneralizedFlowDeterminism2007,mhallaShadowPauliFlow2025}.
We consider gflow here because it is a good intermediate balance point: it is significantly more powerful and interesting than the simpler causal flow, yet less complex to reason about than the computationally universal case of the more general Pauli flow \cite{browneGeneralizedFlowDeterminism2007} or Shadow Pauli flow \cite{mhallaShadowPauliFlow2025}.

Previous work regarding gflow has considered how to find gflow given a labelled open graph \cite{mhallaFindingOptimalFlows2008a,backensThereBackAgain2021,mitosekAlgebraicInterpretationPauli2026}.
The most efficient algorithm for finding gflow is the suitable restriction of the more general algorithm for finding Pauli flow, which runs in $\bigO(n^3)$ \cite{mitosekAlgebraicInterpretationPauli2026,Mitosek2025thesis}.
Moreover, there has been a lot of interest in the family of operations that transform one labelled open graph which has flow into another open graph with flow such that the two represent the same computation, which is used for optimisation \cite{duncanGraphtheoreticSimplificationQuantum2020, backensThereBackAgain2021, simmonsRelatingMeasurementPatterns2021, boothMeasurementbasedQuantumComputation2022, staudacherReducing2QuBitGate2023, staudacherMulticontrolledPhaseGate2024, holkerCausal2024}.
In this context, the existence of flow is also required in order to have an efficient algorithm that translates a one-way computation to a circuit \cite{duncanGraphtheoreticSimplificationQuantum2020,backensThereBackAgain2021,simmonsRelatingMeasurementPatterns2021}.

Yet in the noisy intermediate-scale quantum computing (NISQ) setting, a different question also arises: given some graph state that can be prepared experimentally, what deterministic one-way computations could this state be used to implement \cite{adcockMappingGraphState2020}?
Fully resolving this question is likely computationally hard as there are many different labelled open graphs with the same underlying graph, and it seems unlikely there is an efficient way of checking the existence of gflow for all of them.
Nevertheless, as we show in this work, it is possible to take significant steps towards addressing this question while remaining efficient.

Indeed, we give a polynomial-time algorithm for the following generalisation of the usual flow-finding problem:
\begin{description}[noitemsep]
     \item[Input:] An open graph with no measurement labelling or with some partial constraints on the labelling.
     \item[Output:] A measurement labelling compatible with the constraints, and with the property that the resulting labelled open graph has gflow, or a message that no such labelling exists.
\end{description}
This problem has not been resolved before, though variants of it have been considered.
In \cite{boothOutcomeDeterminismMeasurementbased2023}, Booth et al.\ briefly explain how a solution could function by adapting their qudit flow-finding algorithm; yet that approach is incomplete.
In \cite{mitosekPauliFlowOpen2024, Mitosek2025thesis}, one of the authors showed that, given an open graph, one can efficiently decide whether there exists a measurement labelling consistent with the (more general) Pauli flow.
However, the method used there considered only measurements in the two Pauli eigenbases $X$ and $Z$, which are insufficient for universal computation on graph states; and solutions of this form do not imply solutions to our problem of finding a measurement labelling compatible with gflow.
Flow-finding where the labelling can change during the search was also considered in \cite{erneExtractionZXdiagramsGflowfixed}, which presents a fixed-parameter tractable algorithm for finding a minimal set of vertices that need to be turned into outputs to allow gflow.

Our algorithm for finding a labelling consistent with gflow matches the complexity of the fastest general-purpose gflow-finding algorithm from \cite{mitosekAlgebraicInterpretationPauli2026}.
Since our approach also works if the labelling constraints fully specify a fixed measurement labelling, it is a strict generalisation of the previous algorithm.
Moreover, the new algorithm can take in constraints regarding the partial (time) order of the different measurements and will then return only labellings for which there exists a gflow compatible with these constraints; this does not affect the complexity.
Such constraints may for example arise in experiments, based on the order in which qubits in the graph state are prepared.

It is well-known that any labelled open graph which has gflow also has a `focused gflow' where the corrections for undesired measurement outcomes satisfy certain additional properties \cite{mhallaWhichGraphStates2014a,backensThereBackAgain2021}.
Focused gflow is particularly nice to work with because on labelled open graphs that have equal numbers of inputs and outputs, which we will call \emph{balanced}, there exists at most one focused gflow \cite{mhallaWhichGraphStates2014a,backensThereBackAgain2021,mitosekAlgebraicInterpretationPauli2026}.
We generalise this result to show that on balanced open graphs, there exists at most one measurement labelling that leads to gflow, meaning the useful uniqueness property still holds when no labelling is specified.
Additionally, this motivates a further problem, which we also resolve with a polynomial-time algorithm:
\begin{description}[noitemsep]
    \item[Input:] An unbalanced open graph which has gflow for some labelling.
    \item[Output:] A subset of non-inputs that can be turned into inputs in order to balance the open graph, while preserving the existence of gflow.
\end{description}
This algorithm first finds a labelling for which the input open graph has gflow if no such labelling is specified.
It then makes only minimal changes to this labelling during the balancing procedure: the only labels that may change are those of vertices that become inputs.

This paper is structured as follows.
Section~\ref{sec:preliminaries} contains background and preliminaries.
The new gflow-finding algorithm for open graphs or partial labelling information is developed and proved correct in Section~\ref{sec:algorithm}.
Results related to balanced open graphs and the balancing algorithm are given in Section~\ref{sec:balancing}.
Finally, Section~\ref{sec:conclusions} contains the conclusions.

\section{Preliminaries}
\label{sec:preliminaries}

In the one-way model of measurement-based quantum computation, a computation proceeds via successive adaptive measurements on a resource graph state.
The resource is sometimes called an `open graph state' when input qubits in some arbitrary initial state may be entangled into the graph state itself.
Individual measurements are non-deterministic, but a well-designed one-way computation can be made deterministic overall by adapting later measurements depending on the outcomes of earlier ones.
This exploits the stabiliser property of graph states: every graph state is the joint eigenstate of certain tensor products of Pauli operators.
The permitted measurements are chosen so that the `undesired' measurement outcome can be changed into the desired outcome via one of the Pauli operations.
Hence, a measurement can be corrected after the fact as follows: complete the Pauli operation that changes the undesired outcome into the desired one to a graph stabiliser.
This stabiliser must act non-trivially only on the qubit to be corrected and on qubits that have not yet been measured.
Then the correction `by-products' on the unmeasured qubits can be incorporated into those later measurements; only those Pauli by-products on outputs need to be applied explicitly.

To decide whether a one-way computation can be implemented deterministically, it suffices to know the underlying graph, which qubits are inputs or outputs of the computation, and the measurement label of each non-output qubit.
The measurements that are correctable by Pauli operations are those that lie in the three planes of the Bloch sphere spanned by the eigenstates of two Pauli matrices:
\begin{itemize}
	\item the $\XY$-plane containing states of the form $\ket{\pm_{\XY,\alpha}} = \frac{1}{\sqrt{2}}\left(\ket{0} \pm e^{i\alpha}\ket{1}\right)$,
	\item the $\XZ$-plane containing states of the form $\ket{+_{\XZ,\alpha}} = \cos\frac{\alpha}{2}\ket{0} + \sin\frac{\alpha}{2}\ket{1}$ and $\ket{-_{\XZ,\alpha}} = \sin\frac{\alpha}{2}\ket{0} - \cos\frac{\alpha}{2}\ket{1}$, and
	\item the $\YZ$-plane containing states of the form $\ket{+_{\YZ,\alpha}} = \cos\frac{\alpha}{2}\ket{0} + i\sin\frac{\alpha}{2}\ket{1}$ and $\ket{-_{\YZ,\alpha}} = \sin\frac{\alpha}{2}\ket{0} - i\cos\frac{\alpha}{2}\ket{1}$,
\end{itemize}
where in each case $\alpha$ is an arbitrary real parameter whose value does not affect whether the computation is deterministic (although of course it does affect the specific linear map that is implemented).
Here, a state labelled `$+$' denotes the desired outcome of a measurement and the corresponding state labelled `$-$' denotes the undesired outcome of the same measurement.

A deterministic computation is then witnessed by a flow, which consists of a `correction function' that indicates a correcting stabiliser for each measured qubit, as well as a strict partial order that restricts the order of measurements (to ensure that correction by-products appear only on unmeasured qubits).
We now make these notions formal.

\subsection{Fundamental definitions}

\begin{definition}
    An \emph{open graph} is a triple $\GIO$ where $G = (V,E)$ is a finite simple undirected graph, $I \subseteq V$ is a set of input vertices, and $O \subseteq V$ is a set of output vertices.
\end{definition}

The corresponding resource state used in the computation in the MBQC model is the associated graph state where vertices corresponding to inputs are not initialised to $\ket{+}$:

\begin{definition}\label{def:state}
    Let $(G,I,O)$ be an open graph. We define the corresponding \emph{graph state with input} as $\ket{G,\psi_I}$ as:\begin{equation*}
        \ket{G,\psi_I} := \prod_{uv \in E} CZ_{u,v} \left(\ket{\psi}_I \otimes \bigotimes_{v \in \comp{I} } \ket{+}\right)
    \end{equation*}
    \noindent where $\ket{\psi}_I$ is the input state on qubits in $I$.
\end{definition}

Much of our work concerns the dichotomy of cases depending on whether the numbers of inputs and outputs match. We give the corresponding property a name:

\begin{definition}
    An open graph $(G,I,O)$ is \emph{balanced} when $\abs{I}=\abs{O}$ and \emph{unbalanced} when $\abs{I}<\abs{O}$.
\end{definition}

Note that open graphs with $\abs{I}>\abs{O}$ cannot have gflow, so they will not be considered in this work.

\begin{notation}
    Let $K$ be a stabiliser of a graph state $\ket{G}$.
    We denote the set of vertices on which $K$ has a non-trivial effect as $\supp{K}$, and the set of vertices on which $K$ acts precisely via the Pauli matrix $\sigma$ as $\suppP{\sigma}{K}$ for $\sigma \in \{ \X, \Y, \Z \}$.
\end{notation}

\begin{definition}
    An \emph{open graph stabiliser} $K$ of an open graph $\GIO$ is a stabiliser of the graph state $\ket{G}$ such that $(\suppX{K} \cup \suppY{K}) \cap I = \emptyset$, \ie $K$ can act on inputs only via the Pauli $\Z$.
\end{definition}

Throughout, $n := |V|$.
Since we consider only planar measurements in this work, the following definition is more restrictive than equivalents used elsewhere in the literature in the context of Pauli flow.

\begin{definition}
    A \emph{\LOG} is a quadruple $\lGIO$ where $\GIO$ is an open graph and $\ld \colon \comp{O} \to \{ \XY, \YZ, \XZ \}$ is a \emph{measurement labelling}.
\end{definition}

Before defining flow, we need notation for the odd neighbourhood of a set of vertices, which consist of all those vertices that have an odd number of neighbours in the original set.

\begin{definition}
    Let $G$ be a graph and let $A \subseteq V$ be a set of vertices. The \emph{odd neighbourhood} of $A$ is the set of vertices in $V$ with an odd number of neighbours in $A$:
    \[
    	\odd{A} := \left\{ v \in V : \abs{\{ a \in A \mid av \in E \}} \equiv 1 \bmod 2 \right\}.
    \]
\end{definition}

The following was originally defined in \cite{browneGeneralizedFlowDeterminism2007}; the version below is adapted from \cite{backensThereBackAgain2021}.

\begin{definition}\label{def:gflow}
    A \LOG $\Gamma = (G,I,O,\ld)$ has \emph{generalised flow} (or \emph{gflow} for short) if there exists a map $c \colon \comp{O} \to \powerset(\comp{I})$ and a strict partial order $\prec$ on $V$ such that for all $v \in \comp{O}$ and $w \in V$:\begin{enumerate}
        \item[\conlab{g1}] If $w \in c(v)$ and $v \ne w$, then $v \prec w$.
        \item[\conlab{g2}] If $w \in \odd{c(v)}$ and $v \ne w$, then $v \prec w$.
        \item[\conlab{g3}] If $\ld(v) = \XY$, then $v \notin c(v)$ and $v \in \odd{c(v)}$.
        \item[\conlab{g4}] If $\ld(v) = \YZ$, then $v \in c(v)$ and $v \notin \odd{c(v)}$.
        \item[\conlab{g5}] If $\ld(v) = \XZ$, then $v \in c(v)$ and $v \in \odd{c(v)}$.
    \end{enumerate}
\end{definition}

The map $c$ is called the \textit{correction function}.
For all $v \in \comp{O}$, $c(v)$ is a subset of non-inputs.
This subset is called the \textit{correction set} of $v$.

\begin{definition}\label{def:induced stabiliser}
	Let $\GIO$ be an open graph, and let $C \subseteq \comp{I}$.
	This set induces an open graph stabiliser $K = \prod_{u \in C} \left( X_u \prod_{w \in N(u)} Z_w \right)$, where the subscript denotes the qubit on which the Pauli operator is acting.
\end{definition}

\begin{observation}\label{obs:induced stabiliser}
    In the above definition, by collecting the operators acting on each qubit, we may rewrite $K = \pm \prod_{u\in C} X_u \prod_{w\in\odd{C}} Z_u$, where the sign is irrelevant to our arguments.
    Thus:
	\begin{itemize}
		\item $\suppX{K} = C\setminus\odd{C}$,
		\item $\suppY{K} = C\cap\odd{C}$, and
		\item $\suppZ{K} = \odd{C}\setminus C$.
	\end{itemize}
	The property $I\cap C = \emptyset$ implies this is indeed an \emph{open} graph stabiliser as $I\cap \suppX{K} = \emptyset = I\cap\suppY{K}$.
\end{observation}

\subsection{Simplifying the partial order}

Instead of defining a gflow as a tuple $(c,\prec)$, it is also possible to define a gflow using only the correction function $c$ and the relation that $c$ induces via conditions \conref{g1} and \conref{g2} \cite{broadbentParallelizingQuantumCircuits2009,perdrixDeterminismComputationalPower2017}.
Note that some of the literature on gflow defines the partial order $\prec$ only on the non-outputs, making use of the property that when outputs are included in the partial order, they are always maximal (see Observation~\ref{obs:inputs-minimal-outputs-maximal} below).
We do include outputs in the partial order, which means having to adapt some of the existing results from the literature; this does not cause any problems.

\begin{definition}[{adapted from \cite[Definition~2.11]{mitosekAlgebraicInterpretationPauli2026}}]\label{def:induced-order}
	Let $\Gamma = (G,I,O,\lambda)$ be a labelled open graph and let $c \colon \comp{O} \to \powerset(\comp{I})$ be a correction function on $\Gamma$.
	The \textit{induced relation} $\trl_c$ is the minimal relation on $V$ implied by \conref{g1} and \conref{g2}.
	That is, for all $v \in \comp{O}$ and $w \in V$, the relation $v \trl_c w$ holds if and only if at least one of the following is satisfied:
	\begin{itemize}
		\item $w \in c(v) \wedge v \ne w$ (corresponding to \conref{g1}),
		\item $w \in \odd{c(v)} \wedge v \ne w$ (corresponding to \conref{g2}).
	\end{itemize}
	We denote the transitive closure of $\trl_c$ by $\prec_c$ and call it the \emph{induced partial order}.
\end{definition}

The induced order has useful properties; in particular, it is the minimal order that makes $c$ into a gflow, in the sense of the following two results.

\begin{lemma}[{adapted from \cite[Lemma~2.13]{mitosekAlgebraicInterpretationPauli2026}}]\label{lem:minimalorder containment}
	Let $\Gamma = (G,I,O,\lambda)$ be a labelled open graph and let $(c,\prec)$ be a gflow.
	Then ${\prec_c} \subseteq {\prec}$.
\end{lemma}

\begin{theorem}[{adapted from \cite[Theorem~2.14]{mitosekAlgebraicInterpretationPauli2026}}]\label{thm:extending correction set to Pauli flow}
	Let $\Gamma = (G,I,O,\lambda)$ be a labelled open graph and let $c \colon \comp{O} \to \powerset(\comp{I})$ be a correction function on $\Gamma$.
	Then there exists $\prec$ such that $(c,\prec)$ is a gflow if and only if $(c, \prec_c)$ is a gflow.
\end{theorem}

Since $v \in c(v)$ for all $v$ with $\ld(v) \in \{ \YZ, \XZ \}$, the inputs must instead be measured in the $\XY$ plane or be outputs:

\begin{remark}\label{rem:inputs-XY}
    If a labelled open graph $(G,I,O,\ld)$ has gflow, then $\ld(i) = \XY$ for all $i \in I \setminus O$.
\end{remark}

The computation specified by the labelled open graph can be performed in a robustly deterministic way if and only if the labelled open graph has gflow \cite[Theorem~2 and Theorem~3]{browneGeneralizedFlowDeterminism2007}.
The idea of using gflow for the correction procedure is as follows: Suppose that, when measuring a non-output $v \in \comp{O}$, an undesired outcome is observed.
The corresponding Pauli matrix can be extended to an open graph stabiliser $K$ specified by $c(v)$ in the sense of Definition~\ref{def:induced stabiliser}: in particular, $\suppX{K} = c(v) \setminus \odd{c(v)}$, $\suppZ{K} = \odd{c(v)} \setminus c(v)$, and $\suppY{K} = c(v) \cap \odd{c(v)}$.

For example, if $\ld(v) = \YZ$, the undesired measurement outcome is equivalent to having an undesired Pauli operation $\X$ on $v$. Condition~\ref{g4} ensures that the local effect of the correcting stabiliser $K$ at $v$ is also $\X$. Thus, multiplying by $K$ cancels the $\X$ on the measured vertex and replaces it instead with what is called a `Pauli by-product' on the other vertices in the support of $K$. Conditions~\ref{g3} and~\ref{g5} ensure similar conditions for $\XY$-~and $\XZ$-measured vertices.
Finally, conditions~\ref{g1} and~\ref{g2} ensure that this correction procedure is causal: all vertices which receive a non-trivial Pauli by-product during the correction of $v$ must either be outputs or be measured strictly later than $v$.

\subsection{Focused gflow}

If a labelled open graph has gflow, it generally exhibits multiple different gflows with different correction functions.
One way of reducing the number of different options is called `focusing', which here means restricting the by-products of a correction on other measured vertices depending on the measurement labels of these other vertices.
Focused gflow is often particularly useful to work with (e.g.\ when considering flow-preserving rewrite rules for optimisation) because for balanced labelled open graphs, the focused correction function (if it exists) is unique \cite{mhallaWhichGraphStates2014a,simmonsRelatingMeasurementPatterns2021}.
In the absence of constraints on the partial order or on the existence of qubits at a certain `time', using focused gflow is without loss of generality \cite{mhallaWhichGraphStates2014a,backensThereBackAgain2021}.
The definition of focusing additionally gives rise to certain stabilisers with special properties generated by what are called `focused sets'; these will be important when we address how to balance (labelled) open graphs in Section~\ref{sec:balancing}.

\begin{definition}[{Adapted from \cite[Definition~4.3]{simmonsRelatingMeasurementPatterns2021}}]\label{def:focusing}
	Given a \LOG $\Gamma = (G,I,O,\ld)$, a set $\focusedset\sse\comp{I}$ is \emph{focused over} $S\sse\comp{O}$ if:
	\begin{enumerate}
		\item[\conlab{FX}] For all $w\in S\cap\focusedset$, we have $\ld(w) = \XY$.
		\item[\conlab{FZ}] For all $w\in S\cap\odd{\focusedset}$, we have $\ld(w)\in\{\XZ, \YZ\}$.
	\end{enumerate}
	A \emph{focused set} $\focusedset$ for $\Gamma$ is focused over $\comp{O}$.
	A gflow $(g,\prec)$ is \emph{focused} if $g(v)$ is focused over $\comp{O}\setminus\{v\}$ for all $v\in\comp{O}$.
\end{definition}

The focused sets of a \LOG with gflow form a vector space with dimension $\abs{O} - \abs{I}$.
This space is the kernel of the `flow-demand matrix' \cite[Theorem 3.23]{mitosekAlgebraicInterpretationPauli2026} in the algebraic formulation of gflow, which we present next because it gives rise to the most efficient known gflow-finding algorithm and will also motivate our labelling-finding algorithm in Section~\ref{sec:algorithm}.

\begin{definition}[Flow-demand matrix, adapted from {\cite[Definition~3.4]{mitosekAlgebraicInterpretationPauli2026}}]\label{def:flow-demand matrix}
    Let $\Gamma = (G,I,O,\lambda)$ be a labelled open graph. We define the \emph{flow-demand matrix} $M_{\Gamma}$ as the $(n-n_O) \times (n-n_I)$ matrix with rows corresponding to non-outputs $\comp{O}$ and columns corresponding to non-inputs $\comp{I}$, where the row $M_{v,*}$ corresponding to the vertex $v \in \comp{O}$ satisfies the following for any $w \in \comp{I} \setminus \{ v \}$:\begin{itemize}
        \item if $\lambda(v) = \XY$, then $M_{v,v} = 0$ (provided entry $M_{v,v}$ exists) and $M_{v,w} = Adj_{v,w}$ i.e.\ the $v$ row encodes the neighbourhood of $v$ and
        \item if $\lambda(v) \in \{YZ, XZ\}$, then $M_{v,v} = 1$ (provided entry $M_{v,v}$ exists) and $M_{v,w} = 0$, i.e.\ the $v$ row contains a $1$ at the intersection with the $v$ column and is identically $0$ otherwise.
    \end{itemize}
\end{definition}

\begin{definition}[Order-demand matrix, adapted from {\cite[Definition~3.5]{mitosekAlgebraicInterpretationPauli2026}}]\label{def:order-demand matrix}
     Let $\Gamma = (G,I,O,\lambda)$ be a labelled open graph. We define the \textit{order-demand matrix} $N_{\Gamma}$ as the $(n-n_O) \times (n-n_I)$ matrix with rows corresponding to non-outputs $\comp{O}$ and columns corresponding to non-inputs $\comp{I}$, where the row $N_{v,*}$ corresponding to the vertex $v \in \comp{O}$ satisfies the following for any $w \in \comp{I}\setminus \{ v \}$:\begin{itemize}
        \item if $\lambda(v) = YZ$, then $N_{v,v} = 0$ (provided entry $N_{v,v}$ exists) and $N_{v,w} = Adj_{v,w}$, i.e.\ the $v$ row encodes the neighbourhood of $v$,
        \item if $\lambda(v) = XZ$, then $N_{v,v} = 1$ (provided entry $N_{v,v}$ exists) and $N_{v,w} = Adj_{v,w}$, i.e.\ the $v$ row encodes the neighbourhood of $v$ and also has a $1$ at the intersection with the $v$ column, and
        \item if $\lambda(v) = XY$, then $N_{v,v} = 1$ (provided entry $N_{v,v}$ exists) and $N_{v,w} = 0$, i.e.\ the $v$ row contains a $1$ at the intersection with the $v$ column and is identically $0$ otherwise.
    \end{itemize}
\end{definition}

\begin{definition}[Correction matrix, adapted from {\cite[Definition~3.6]{mitosekAlgebraicInterpretationPauli2026}}]\label{def:correction-matrix}
    Let $\Gamma = (G,I,O,\lambda)$ be a labelled open graph and let $c$ be a correction function on $\Gamma$. We define the \textit{correction matrix $C$} encoding the function $c$ as the $(n-n_I) \times (n-n_O)$ matrix with rows corresponding to non-inputs $\comp{I}$ and columns corresponding to non-outputs $\comp{O}$, where $C_{u,v} = 1$ if and only if $u \in c(v)$.
\end{definition}

\begin{theorem}[Algebraic formulation of gflow, adapted from {\cite{mitosekAlgebraicInterpretationPauli2026}}]\label{th:algebraic formulation}
    Let $\Gamma = (G,I,O,\lambda)$ be a labelled open graph, $c$ be a correction function, $M$ be the flow-demand matrix of $\Gamma$, and $N$ be the order-demand matrix of $\Gamma$. Then $\Gamma$ has gflow if and only if there exists a correction matrix $C$ such that:\begin{itemize}
        \item $MC = Id_{\comp{O}}$,
        \item $NC$ is the adjacency matrix of a directed acyclic graph.
    \end{itemize}
\end{theorem}

For an example, see Figure~\ref{fig:uniqueness-example}.

\begin{figure}
    \centering

    \begin{subfigure}[c]{0.36\textwidth}
        \centering
        \begin{tikzpicture}
	\begin{pgfonlayer}{nodelayer}
		\node [style=ZH H] (18) at (-3, 0) {};
        \node [style=ZH H] (19) at (2, 0) {};
		\node [style=GR Tiny Black] (0) at (-3, 0) {};   
		\node [style=GR Tiny Black] (1) at (-1, 2) {};   
		\node [style=GR Tiny Black] (2) at (-1, -2) {};  
		\node [style=GR Tiny Black] (3) at (0, 0) {};    
		\node [style=GR Tiny Black] (4) at (2, 0) {};    
		\node [style=GR Tiny Empty] (5) at (4, 1) {};    
		\node [style=GR Tiny Empty] (6) at (4, -2) {};   
		\node [style=BlackTEXT] (7) at (-3, -0.6) {$i^{\XY}$};
		\node [style=BlackTEXT] (8) at (-1, 2.5) {$a^{\XZ}$};
		\node [style=BlackTEXT] (9) at (-1, -2.5) {$b^{\XY}$};
		\node [style=BlackTEXT] (10) at (0.2, -0.5) {$c^{\YZ}$};
		\node [style=BlackTEXT] (11) at (2, 0.6) {$j^{\XY}$};
		\node [style=BlackTEXT] (12) at (4, 1.5) {$o_1$};
		\node [style=BlackTEXT] (13) at (4, -2.5) {$o_2$};
		\node [style=none] (14) at (-4.5, 0) {};
		\node [style=none] (15) at (4.5, 0) {};
		\node [style=none] (16) at (0, 3.5) {};
		\node [style=none] (17) at (0, -3.5) {};
	\end{pgfonlayer}
	\begin{pgfonlayer}{edgelayer}
		\draw (0) to (1); 
		\draw (0) to (2); 
		\draw (1) to (2); 
		\draw (1) to (3); 
		\draw (1) to (4); 
		\draw (2) to (3); 
		\draw (2) to (4); 
		\draw (3) to (4); 
		\draw (4) to (5); 
		\draw (2) to (6); 
	\end{pgfonlayer}
\end{tikzpicture}
        \caption{The labelled open graph.}
        \label{fig:uniqueness-example-graph}
    \end{subfigure}
    \hfill
    \begin{subfigure}[c]{0.33\textwidth}
        \centering
        \[
        \scalemath{0.85}{
        \begin{pNiceArray}{cc:ccc}[first-row, first-col]
                & i & j & a & b & c \\
            a   & 0 & 0 & 1 & 0 & 0 \\
            b   & 1 & 0 & 1 & 0 & 0 \\
            c   & 0 & 0 & 0 & 0 & 1 \\
            \hdottedline
            o_1 & 1 & 1 & 0 & 0 & 1 \\
            o_2 & 0 & 0 & 1 & 1 & 1
        \end{pNiceArray}}
        \]
        \caption{The correction matrix.}
        \label{fig:uniqueness-example-correction}
    \end{subfigure}
    \hfill
    \begin{subfigure}[c]{0.28\textwidth}
        \centering
        \[
        \scalemath{0.85}{
        \begin{pNiceArray}{cc:ccc}[first-row, first-col]
                & i & j & a & b & c \\
            i   & 0 & 0 & 0 & 0 & 0 \\
            j   & 0 & 0 & 0 & 0 & 0 \\
            \hdottedline
            a   & 1 & 0 & 0 & 0 & 1 \\
            b   & 1 & 0 & 1 & 0 & 0 \\
            c   & 1 & 0 & 0 & 0 & 0
        \end{pNiceArray}}
        \]

        \caption{The induced relation matrix.}
        \label{fig:uniqueness-example-order}
    \end{subfigure}

    \caption{\hyperref[fig:uniqueness-example-graph]{(\subref*{fig:uniqueness-example-graph})} An open graph with $I = \{ i, j \}$ and $O = \{ o_1, o_2 \}$, together with a measurement labelling resulting in gflow; \hyperref[fig:uniqueness-example-correction]{(\subref*{fig:uniqueness-example-correction})} a correction matrix corresponding to focused gflow; \hyperref[fig:uniqueness-example-order]{(\subref*{fig:uniqueness-example-order})} and a matrix of the corresponding induced relation restricted to non-outputs.}
    \label{fig:uniqueness-example}
\end{figure}

\begin{proposition}[{\cite[Prop.~3.14]{backensThereBackAgain2021}}]\label{prop:focused}
	A \LOG $(G,I,O,\ld)$ has gflow if and only if it has focused gflow.
\end{proposition}

\begin{observation}\label{obs:inputs-minimal-outputs-maximal}
	Let $\Gamma = (G,I,O,\lambda)$ be a labelled open graph with gflow $(c,\prec_c)$ where $c \colon \comp{O} \to \powerset(\comp{I})$ is a correction function on $\Gamma$ and $\prec_c$ is the corresponding induced order.
	Then all outputs are maximal in $\prec_c$: outputs have no correction sets and thus do not appear on the left-hand side of $\trl_c$, so they do not appear on the left-hand side of the transitive closure of $\trl_c$ (\ie $\prec_c$).
	
	Moreover, if $c$ is focused, then all inputs that are not outputs are minimal in $\prec_c$.
	To see this, recall that inputs can only be $\XY$-measured by Remark~\ref{rem:inputs-XY}.
	Now, inputs do not appear in correction sets by the definition of a correction function, and $\XY$-measurements do not appear in odd neighbourhoods of correction sets (other than their own) by focusing.
	Hence inputs cannot appear on the right-hand side of $\trl_c$ unless they are also outputs.
\end{observation}

\begin{remark}
	As a consequence of Proposition~\ref{prop:focused} and Observation~\ref{obs:inputs-minimal-outputs-maximal}, if a partial order is not fixed, we may assume without loss of generality that inputs which are not outputs are minimal, and that outputs are maximal.
\end{remark}

The above properties of the partial order also follow directly from the algebraic formulation:

\begin{remark}
    Let $(G,I,O,\ld)$ be a \LOG, $M$ the corresponding flow-demand matrix, $N$ the corresponding order-demand matrix, $(c,\prec)$ a focused gflow and $C$ the correction matrix corresponding to $c$. Then the product $NC$ is exactly the matrix of $\trl_c$ restricted to $\comp{O}$, that is $\forall v,w \in \comp{O}$ we have $\left(NC\right)_{v,w} = 1 \Leftrightarrow w \trl_c v$.
\end{remark}

\subsection{Flow preservation}

Some of our proofs invoke flow-preserving rewrite rules: these are transformations that modify the \LOG in such a way that if the \LOG prior to the transformation had gflow, then so does the \LOG afterwards.
These rules are interesting in their own right, particularly for their applications in optimising quantum computations; see for instance \cite{duncanGraphtheoreticSimplificationQuantum2020,backensThereBackAgain2021,staudacherReducing2QuBitGate2023,mcelvanneyCompleteFlowPreservingRewrite2023,mcelvanneyFlowpreservingZXcalculusRewrite2023,mitosekAlgebraicInterpretationPauli2026,perezBackens2025,backensCompleteness2026}.
Here, we include only the few needed for our proofs.

\begin{theorem}[Local complementation; adapted from {\cite[Definition~2.24 and~Lemmata~3.1 and~3.2]{backensThereBackAgain2021}}]\label{th:LC}
    Let $(G,I,O,\ld)$ be a \LOG with gflow and let $u \in \comp{I}$. Then $(G*u,I,O,\ld')$ has gflow where:\begin{equation*}
        G*u = (V(G),E(G) \symd \{ (b,c) \mid (b,u), (c,u) \in E(G) \wedge b \ne c \}),
    \end{equation*}
    and $\ld'(u)$ is defined provided $u \notin O$ as follows:
    \begin{equation*}
        \ld'(u) = \begin{cases}
            \XZ &\text{if } \ld(u)=\XY \\
            \XY &\text{if } \ld(u)=\XZ \\
            \YZ &\text{if } \ld(u)=\YZ
        \end{cases}
    \end{equation*}
    and for $v \in \comp{O} \setminus \{ u \}$:\begin{equation*}
        \ld'(v) = \begin{cases}
            \YZ &\text{if } v \in N_G(u) \wedge \ld(v)=\XZ \\
            \XZ &\text{if } v \in N_G(u) \wedge \ld(v)=\YZ \\
            \ld(v) &\text{otherwise}.
        \end{cases}
    \end{equation*}
\end{theorem}

\begin{lemma}[Removal of $\XZ$ and $\YZ$ labelled vertices, adapted from {\cite[Lemma~3.4]{backensThereBackAgain2021}}]\label{lemma:Z-like removal}
    Let $(G,I,O,\ld)$ be a \LOG with gflow and let $u \in \comp{O}$ with $\ld(u) \in \{ \XZ, \YZ)$. Then $(G-u,I,O,\ld\rvert_{\comp{O}\setminus\{u\}})$ has gflow where $G-v$ is graph $G$ with vertex $v$ removed.
\end{lemma}

\begin{observation}[Insertion of $\XZ$ or $\YZ$ labelled vertex, adapted from {\cite[Observation~4.8 and Theorems~4.3 and~5.2]{perezBackens2025}}]\label{obs:Z-like insertion}
    Let $(G,I,O,\ld)$ be a \LOG with gflow where $\abs{I}=\abs{O}$ and let $S \subseteq V$ be a subset of vertices. Then at most one of $(G',I,O,\ld')$ and $(G',I,O,\ld'')$ has gflow where:\begin{equation*}
        G' = (V\sqcup\{v_S\}, E(G) \cup \{ v_Su \mid u \in S \})
    \end{equation*}
    and the new measurement labellings are:
    \begin{equation*}
        \ld'(u) = \begin{cases}
            \ld(u) &\text{if } u\in V\setminus O \\
            \XZ &\text{if } u = v_S
        \end{cases}
        \qquad\text{and}\qquad
        \ld''(u) = \begin{cases}
            \ld(u) &\text{if } u\in V\setminus O \\
            \YZ &\text{if } u = v_S.
        \end{cases}
    \end{equation*}
\end{observation}

\section{Algorithm for finding a measurement labelling}\label{sec:algorithm}

First, we make formal the informal description in the introduction of finding a measurement labelling that is consistent with gflow.
It will be useful to modify the definition of a gflow on a labelled open graph (Def.~\ref{def:gflow}) by moving the labelling $\lambda$ from the input open graph tuple (which is then no longer labelled) to the output gflow tuple instead.

\begin{definition}\label{def:lgflow}
    A \emph{labelled gflow} (or \emph{lgflow} for short) on an open graph $\GIO$ is a triple $(\ld,c,\prec)$ where:\begin{itemize}
        \item $\ld \colon \comp{O} \to \{ \XY, \YZ, \XZ \}$ is a measurement labelling on $\GIO$ and
        \item $(c,\prec)$ is a gflow on $\lGIO$.
    \end{itemize} 
\end{definition}

\begin{definition}\label{def:space of labellings}
    Let $\GIO$ be an open graph.
    We say that $\ld$ \emph{results in gflow} if it is part of some lgflow on $\GIO$.
    We call the set of $\ld$ that result in gflow the \emph{\lspace}.
\end{definition}

Now, we can define the main problem studied in this work:
\problemstatement{\mainproblem}{an open graph $\GIO$}{an lgflow $(\ld,c,\prec)$ on $\GIO$, or a message that no such $(\ld,c,\prec)$ exists}
Equivalently, the goal is to find an element of the \lspace, if it is non-empty.

We show that $\mainproblem$ can be solved in $\bigO(n^3)$ where $n = |V|$.
The algorithm utilises a layer-by-layer approach for flow finding, initially used in \cite{mhallaFindingOptimalFlows2008a}.
Specifically, we combine ideas from \cite{backensThereBackAgain2021} (which generalises \cite{mhallaFindingOptimalFlows2008a}) and \cite{mitosekAlgebraicInterpretationPauli2026} to obtain an $\bigO(n^3)$ method.
In particular, our approach uses the earlier layer-by-layer ideas from \cite{backensThereBackAgain2021}, but splits the Gaussian elimination cost more efficiently across many layers, still yielding an $\bigO(n^3)$ bound.
Unlike the algorithms of \cite{backensThereBackAgain2021,mitosekAlgebraicInterpretationPauli2026}, our method does not necessarily produce a focused flow, though the produced flow is still guaranteed to be maximally delayed, meaning qubits are measured as late as possible for compatibility with flow.
(As flow-finding algorithms generally work down the partial order from the outputs, maximally delayed flow corresponds to a greedy algorithmic approach.)

We include a worked example illustrating a run of the algorithm in Appendix~\ref{subsec:worked example balanced} and~\ref{subsec:worked example unbalanced}.
We also briefly discuss relevant optimisation tricks in Appendix~\ref{sec:implementation}.

\subsection{Algorithm overview}\label{sec:algo-overview}
Before presenting the efficient $\bigO(n^3)$ algorithm, we consider a minimal algorithm that takes an open graph $\GIO$ and returns a function $\ld$ such that $\lGIO$ has gflow.
The algorithm follows the steps below:\begin{enumerate}
    \item Initialise an empty $\ld$ and initialise the set of solved vertices $\Solved$ to contain all outputs. If $\Solved = V$, return $\ld$.
    \item For all $v \in V\setminus\Solved$, find all planes in which $v$ can be measured in such a way that all order constraints $v \prec u$ implied by the correction set for $v$ satisfy $u \in \Solved$. If $v \in I$, accept only the $\XY$ plane. If no suitable measurement planes were found for any vertex, return that there is no lgflow. Otherwise, record some valid $\ld(v)$ for each vertex $v$ where suitable measurement planes were found, and add the newly solved vertices to $\Solved$.
    \item Check if $\Solved = V$. If not, repeat the previous step. If yes, return the recorded $\ld$.
\end{enumerate}

As we will prove, the above algorithm is correct. The above algorithm does not keep track of the correction sets, yet once $\ld$ is found, the standard gflow-finding algorithm can be applied to obtain the full lgflow $(\ld, c, \prec)$.

The naive approach to the above algorithm runs in $\bigO(n^5)$ time complexity: in the second step, $\bigO(n)$ vertices are considered, each requiring $\bigO(n^3)$ time to check each possible measurement plane via Gaussian elimination, and the second step can be entered $\bigO(n)$ times.
The full algorithm we present later follows the same general idea, but fuses the construction of $(c,\prec)$ into the algorithm and takes a more careful approach to solving the underlying linear systems, achieving $\bigO(n^3)$ time complexity.

\subsection{Finding the measurement label for a single vertex}
The second step of the algorithm sketched in Section~\ref{sec:algo-overview} attempts to find $\ld(v)$ such that the order constraints implied by the correction set of $v$ only necessitate $v \prec u$ for $u \in \Solved$, \ie the implied `forward cone' (defined below) of $v$ lies entirely in the set of previously solved vertices~$\Solved$.
In this subsection, we show how to achieve this.
First, we formalise the meaning of the forward cone:

\begin{definition}
    Let $\GIO$ be an open graph, let $v \in \comp{O}$ be a vertex, and let $\mathcal{C} \subseteq \comp{I}$ be such that the stabiliser induced by $\mathcal{C}$ has a non-trivial effect on $v$.
    We define the \emph{forward cone of $v$ with respect to $\mathcal{C}$}, denoted $\future{v}{\mathcal{C}}$, as the set of vertices that must come directly after $v$ in the partial order if $v$ has correction set $\mathcal{C}$, \ie \begin{equation*}
        \future{v}{\mathcal{C}} := \{ u \in V\setminus\{ v \} \mid u \in \mathcal{C} \vee u \in \odd{\mathcal{C}} \}
    \end{equation*}
\end{definition}

The above definition adapts \cite[Definition 1]{raussendorfCone2002}, which considered only the $\XY$ planar measurements on cluster states.
Therefore, $\future{v}{\mathcal{C}}$ is precisely the set of vertices $u$ such that the gflow conditions~\ref{g1} and~\ref{g2} imply $v \prec u$ if $\mathcal{C}$ is used as the correction set for $v$; correspondingly it is the set of vertices such that $v\trl_c u$ if $c(v) = \mathcal{C}$, \cf Definition~\ref{def:induced-order}.

To check which measurement plane can be assigned to a given vertex, we use the following:

\begin{observation}
    Let $\GIO$ be an open graph, $v \in \comp{O}$ be a vertex, and $S$ be a set such that $O \subseteq \Solved \subseteq V$. Then $v \in V$ can be measured in $\ld(v)$ with a correction procedure that implies $v \prec u$ only for $u \in \Solved$ if and only if there exists a set $\mathcal{C} \subseteq \comp{I}$ such that $v\in\mathcal{C}\cup\odd{\mathcal{C}}$ and $\future{v}{\mathcal{C}} \subseteq \Solved$.
    Furthermore, $\mathcal{C}$ specifies the correction set of $v$ and it implies $\ld(v)$ as follows:\begin{itemize}
        \item If $v \in \mathcal{C} \setminus \odd{\mathcal{C}}$, then $\ld(v) = \YZ$,
        \item if $v \in \odd{\mathcal{C}} \setminus \mathcal{C}$, then $\ld(v) = \XY$, and
        \item if $v \in \mathcal{C} \cap \odd{\mathcal{C}}$, then $\ld(v) = \XZ$.
    \end{itemize}
\end{observation}

Suppose, given a set $\Solved$, we want to find a correction set $\mathcal{C}$ for vertex $v$ such that the forward cone $\future{v}{\mathcal{C}}$ is contained in $\Solved$ (or conclude that no such $\mathcal{C}$ exists).
To do this, we must solve certain systems of linear equations.

\begin{notation}
    Given a matrix $M$ with its sets of row labels and column labels, we denote the submatrix of $M$ restricted to the rows in set $R$ and columns in set $C$ by $M[R,\ C]$.
    This is sometimes called the $R\times C$ submatrix of $M$.
    
    We also use $*$ as shorthand for all rows or all columns; for instance, $M[*,\ C]$ is the restriction of $M$ to all rows and to those columns in $C$.
    Finally, we use $M_{*,v}$ for the vector $M[*,\ \{v\}]$ and $M_{u,*}$ for the vector $M[\{u\},\ *]$.
\end{notation}

\begin{definition}
    Given an open graph $(G,I,O)$, its \emph{reduced adjacency matrix} $\Adj$ is the $\comp{O} \times \comp{I}$ submatrix of the full adjacency matrix (which has both rows and columns labelled by $V$).
\end{definition}

\begin{definition}
    Given a finite set $W$, the indicator vector of a subset $T\sse W$ is the vector $t\in\Ftwo^W$ that satisfies $t_w = 1 \Leftrightarrow w\in T$ for all $w\in W$.
    We sometimes write $T' = \supp{t}$.
    In a slight abuse of notation, we identify singlet sets in this context with their sole element, \ie we may talk about the indicator vector of an element $a\in W$, which is just the indicator vector of $\{a\}$.
\end{definition}

\begin{lemma}\label{lemma:linear systems}
    Let $\GIO$ be an open graph, $\Solved$ be a set such that $O \subseteq \Solved \subseteq V$, and $a \in V \setminus \Solved$ be a vertex.
    Let $\Adj$ be the reduced adjacency matrix of $G$, \ie the $\comp{O} \times \comp{I}$ submatrix of the full adjacency matrix, and let $e_a \in \Ftwo^{V \setminus \Solved}$ be the indicator vector of $a$.
    Let $A$ be the submatrix of the adjacency matrix and let $b_a$ be the column vector defined as follows:\begin{align*}
        A &:= \Adj[V \setminus \Solved,\ \Solved \setminus I]\\
        b_a &:= \Adj[V \setminus \Solved,\ \{a\}].
    \end{align*}
    Then, a set $\mathcal{C} \subseteq \comp{I}$ such that $\future{a}{\mathcal{C}} \subseteq \Solved$ exists if and only if at least one of the following holds:\begin{enumerate}
        \item The linear system $A\SolutionVector = e_a$ has a solution $\SolutionVector$. Then $\mathcal{C} = \supp{\SolutionVector}$ and $\ld(a) = \XY$.
        \item The vertex $a$ is not an input, and the linear system $A\SolutionVector = b_a$ has a solution $\SolutionVector$. Then $\mathcal{C} = \supp{\SolutionVector} \cup \{a\}$ and $\ld(a) = \YZ$.
        \item The vertex $a$ is not an input, and the linear system $A\SolutionVector = b_a + e_a$ has a solution $\SolutionVector$. Then $\mathcal{C} = \supp{\SolutionVector} \cup \{a\}$ and $\ld(a) = \XZ$.
    \end{enumerate}
\end{lemma}
\begin{proof}
    $(\Rightarrow)$: for every $u \in (V \setminus \Solved) \setminus \{a\}$, the assumption $\future{a}{\mathcal{C}} \subseteq \Solved$ implies $u \notin \mathcal{C} \cup \odd{\mathcal{C}}$.
    We consider the three possible non-trivial effects of the stabiliser induced by $\mathcal{C}$ on $a$.\begin{enumerate}
        \item Suppose $a \in \odd{\mathcal{C}} \setminus \mathcal{C}$.
        Then $\mathcal{C} \subseteq \Solved \setminus I$.
        Let $\SolutionVector \in \Ftwo^{\Solved \setminus I}$ satisfy $\supp{\SolutionVector} = \mathcal{C}$.
        Since $\odd{\mathcal{C}} \cap (V \setminus \Solved) = \{a\}$, we have $A\SolutionVector = e_a$.

        \item Suppose $a \in \mathcal{C} \setminus \odd{\mathcal{C}}$.
        Necessarily $a \notin I$, and $\mathcal{C} \setminus \{a\} \subseteq \Solved \setminus I$.
        Let $\SolutionVector \in \Ftwo^{\Solved \setminus I}$ satisfy $\supp{\SolutionVector} = \mathcal{C} \setminus \{a\}$.
        Since $\odd{\mathcal{C}} \cap (V \setminus \Solved) = \emptyset$, we have $A\SolutionVector + b_a = 0$, which yields the second system.

        \item Suppose $a \in \mathcal{C} \cap \odd{\mathcal{C}}$.
        Again $a \notin I$ and $\mathcal{C} \setminus \{a\} \subseteq \Solved \setminus I$.
        For $\SolutionVector \in \Ftwo^{\Solved \setminus I}$ with $\supp{\SolutionVector} = \mathcal{C} \setminus \{a\}$, the equality $\odd{\mathcal{C}} \cap (V \setminus \Solved) = \{a\}$ yields $A\SolutionVector + b_a = e_a$, which leads to the third system.
    \end{enumerate}

    $(\Leftarrow)$: suppose that one of the three systems has a solution $\SolutionVector$ and construct $\mathcal{C}$ as stated.
    In the first case, $\mathcal{C} \subseteq \Solved \setminus I$ and the equation gives $\odd{\mathcal{C}} \cap (V \setminus \Solved) = \{a\}$, so $a \in \odd{\mathcal{C}} \setminus \mathcal{C}$.
    In the second and third cases, $\mathcal{C} \setminus \{a\} \subseteq \Solved \setminus I$; the respective equations give $\odd{\mathcal{C}} \cap (V \setminus \Solved) = \emptyset$ and $\odd{\mathcal{C}} \cap (V \setminus \Solved) = \{a\}$, so $a \in \mathcal{C} \setminus \odd{\mathcal{C}}$ and $a \in \mathcal{C} \cap \odd{\mathcal{C}}$, respectively.
    In every case, no vertex of $(V \setminus \Solved) \setminus \{a\}$ belongs to $\mathcal{C} \cup \odd{\mathcal{C}}$, and therefore $\future{a}{\mathcal{C}} \subseteq \Solved$.
\end{proof}

This concludes the explanation of the second step of the general approach. The method assigns a label to a vertex as soon as it detects that an assignment is possible, \ie as soon as the set of solved vertices $\Solved$ can contain the entire forward cone of $a$. Note that when attempting to find a label for a specific $a$, the labels of vertices that will be in the forward cone are unimportant; it only matters that these vertices have already been solved. Effectively, this is equivalent to the original three-plane gflow finding procedure \cite{backensThereBackAgain2021}, albeit without ensuring that the recorded set $\mathcal{C}$ corresponds to a focused gflow. The distinction is important: if $\mathcal{C}$ is required to satisfy the focusing conditions, the linear system described above no longer works because the labels of vertices in $\Solved$ then do matter. Nevertheless, once a labelling is found, we can determine a focused flow in $\bigO(n^3)$; thus, our algorithm has no real disadvantage.

We can now move to the presentation of the $\bigO(n^3)$ algorithm.
Here, the algorithm is not optimised beyond the target complexity.
We discuss various improvements in Appendix~\ref{sec:implementation}.

\subsection{Full algorithm}

We will set up a whole family of systems of linear equations that share some of their coefficients but have different constants.
Yet at the beginning, only part of the variables in this system are allowed to take non-zero values; more variables become available in this sense as the algorithm proceeds.

\begin{definition}\label{def:active}
    A column of coefficients is called \emph{active} if the corresponding variable is allowed to take non-zero values at this stage of the algorithm.
\end{definition}

We now show the following:

\begin{theorem}
    There exists an $\bigO(n^3)$ algorithm solving $\mainproblem$, where $n = |V|$.
\end{theorem}

\begin{algorithm}
\caption{Constructing the labelling-finding linear system}\label{algo:system construction}
\begin{algorithmic}[1]
\Statex Returns a linear system.
\Statex The given graph is assumed to have some underlying order on the vertices $\Vertices$. We refer to it as graph-order. It is not the temporal order we are searching for. Rows are non-outputs $\Rows=\Vertices\setminus O$; first-block columns are non-inputs $\Cols=\Vertices\setminus I$. The first block is the reduced adjacency matrix. Attached vectors encode possible labels.
\Procedure{BuildLabellingSystem}{$G,I,O$}\label{proc:BuildLabellingSystem}
    \State $\Vertices \gets V(G)$
    \State $\Rows \gets$ graph-order $(\Vertices\setminus O)$ \Comment{ordered list of non-outputs according to presentation of $\Vertices$}
    \State $\Cols \gets$ graph-order $(\Vertices\setminus I)$ \Comment{ordered list of non-inputs according to presentation of $\Vertices$}
    \State $\Adj \gets$ reduced adjacency matrix of $G$ \Comment{$\Adj$ has shape $\Rows \times \Cols$}
    \State $\Req \gets [\,\,]$ \Comment{this will be the ordered list of potential vertex-label pairs $(v,\lambda)$}
    \State $\AttachedBlock \gets$ empty matrix with $|\Rows|$ rows
    \For{$v \in \Rows$}
        \State $\Labels_v \gets \Call{CandidateLabels}{v,I}$
        \State $\rhs_{v,\XY} \gets$ unit vector with a $1$ in row $v$
        \If{$\XY \in \Labels_v$}
            \State append column $\rhs_{v,\XY}$ to $\AttachedBlock$ and append $(v,\XY)$ to $\Req$
        \EndIf
        \If{$\YZ \in \Labels_v$}
            \State $\rhs_{v,\YZ} \gets \Adj_{*,v}$
            \State append column $\rhs_{v,\YZ}$ to $\AttachedBlock$ and append $(v,\YZ)$ to $\Req$
        \EndIf
        \If{$\XZ \in \Labels_v$}
            \State $\rhs_{v,\XZ} \gets \Adj_{*,v} + \rhs_{v,\XY}$
            \State append column $\rhs_{v,\XZ}$ to $\AttachedBlock$ and append $(v,\XZ)$ to $\Req$
        \EndIf
    \EndFor
    \State $\LS,\ILS \gets \bigl[\,\Adj \mid \AttachedBlock \mid \Identity_{\Rows}\,\bigr]$ \Comment{two independent copies}
    \State \Return $(\LS,\ILS,\Rows,\Cols,\Req)$
\EndProcedure
\Statex
\Statex Returns possible labels for vertex $v$ by checking whether $v$ is an input.
\Procedure{CandidateLabels}{$v,I$}\label{proc:CandidateLabels}
    \If{$v\in I$}
        \State $\Labels_v \gets \{\XY\}$
    \Else
        \State $\Labels_v \gets \{\XY,\YZ,\XZ\}$
    \EndIf
    \State \Return $\Labels_v$
\EndProcedure
\end{algorithmic}
\end{algorithm}

\begin{algorithm}
\caption{Main algorithm}\label{algo:main algo}
\begin{algorithmic}[1]
\Statex Returns a measurement labelling $\ld$ resulting in gflow or a message that no such $\ld$ exists. The starting open graph is assumed to satisfy $I \cap O = \emptyset$.
\Procedure{FindGflowLabelling}{$G,I,O$}\label{proc:FindGFlowLabelling}
    \State $\Vertices \gets V(G)$
    \State $(\LS,\ILS,\Rows,\Cols,\Req) \gets \Call{BuildLabellingSystem}{G,I,O}$
    \State $\activecount \gets 0$ \Comment{number of active first-block columns, initially zero}
    \State $\Solved \gets O$ \Comment{solved vertices, initially only the outputs}
    \State $\Layers \gets [\,O\,]$ \Comment{layers of the temporal order, initially only the outputs}
    \State $\Solutions \gets \emptyset$
    \State $(\LS,\ILS,\Cols,\activecount) \gets \Call{ActivateVertices}{\LS,\ILS,\Cols,\activecount,O}$
    \Statex \Comment{incrementally activates all output columns, cf.\ Definition~\ref{def:active}}
    \State $\Unsolved \gets V \setminus \Solved$ \Comment{unsolved non-outputs}
    \While{$\Solved \ne \Vertices$} \Comment{repeat until all vertices are solved}
        \State $\Found \gets \Call{FindNewSolutions}{\LS,\Req,\Cols,\activecount,\Solved,\Candidates}$
        \If{$\Found=\emptyset$}
            \State \Return No labelling found
        \EndIf
        \State $\New \gets$ vertices of $\Rows$ that appear in $\Found$\label{Line:New from Layers}
        \State append $\New$ to $\Layers$
        \State $(\LS,\ILS,\Req) \gets \Call{DiscardAttachedColumns}{\LS,\ILS,\Req,\New}$
        \For{$v\in\New$}
            \State $\Solutions[v]\gets \Found[v]$
            \State $(\LS,\ILS) \gets \Call{RemoveInitialRow}{\LS,\ILS,v,\activecount}$
            \State $(\LS,\ILS,\Cols,\activecount) \gets \Call{ActivateVertices}{\LS,\ILS,\Cols,\activecount,\{v\}}$
            \State $\Solved \gets \Solved \cup \{v\}$
        \EndFor
        \State $\Unsolved \gets V \setminus \Solved$ \Comment{update unsolved non-outputs}
    \EndWhile
    \State \Return successful result $(G,\Layers,\Solutions,\Solved)$
\EndProcedure
\end{algorithmic}
\end{algorithm}

\begin{algorithm}
\caption{Solving all currently available vertices}\label{algo:layer solving}
\begin{algorithmic}[1]
\Statex Finds solutions for the next layer of vertices.
\Procedure{FindNewSolutions}{$\LS,\Req,\Cols,\activecount,\Solved,\Candidates$}\label{proc:FindNewSolutions}
    \State $\CoefficientBlock \gets (\LS)[*,\ 1..\activecount]$ \Comment{coefficient block (part of the first block), already in row echelon form}
    \State $\ZeroRows \gets \{r \mid (\CoefficientBlock)_{r,*}=\zero\}$ \Comment{zero rows of the coefficient block}
    \State $\Active \gets$ list of first $\activecount$ vertices of $\Cols$ in order
    \State $\Requests\gets[\,]$
    \For{$v\in \Candidates$}
        \For{$\lambda$ such that $(v,\lambda)\in\Req$}
            \State $\rhs \gets$ attached column of $\LS$ labelled $(v,\lambda)$
            \If{$\rhs_r=0$ for every $r\in \ZeroRows$}
                \State append $(v,\lambda,\rhs)$ to $\Requests$
            \EndIf
        \EndFor
    \EndFor
    \If{$\Requests$ is empty}
        \State \Return $\emptyset$
    \EndIf
    \State Solve all systems $\CoefficientBlock\SolutionVector = \rhs$ for requested right-hand sides by back substitution\label{line:Ex=b}
    \State $\Found\gets\emptyset$
    \For{each solved request $(v,\lambda,\rhs)$ with solution $\SolutionVector$}
        \State $\CorrectionSet \gets \{\,\Active_j \mid \SolutionVector_j=1\,\}$
        \If{$\lambda\in\{\YZ,\XZ\}$}
            \State $\CorrectionSet \gets \CorrectionSet\cup\{v\}$
        \EndIf
        \State store witness $(\lambda,\CorrectionSet,\Solved,\Active,\SolutionVector)$ in $\Found[v]$
    \EndFor
    \State \Return $\Found$
\EndProcedure
\end{algorithmic}
\end{algorithm}

\begin{algorithm}
\caption{Single-row maintenance of the labelling linear system}\label{algo:row maintenance}
\begin{algorithmic}[1]
\Statex Discards attached vectors that belong to vertices in the newly solved layer. The corresponding labelled columns in the second block are removed from both systems (whereas the first- and third-block columns associated with those vertices are kept).
\Procedure{DiscardAttachedColumns}{$\LS,\ILS,\Req,\ActivationSet$}\label{proc:DiscardAttachedColumns}
    \State delete from $\LS$ and $\ILS$ every attached column labelled $(v,\lambda)$ with $v\in\ActivationSet$
    \State $\Req \gets [\,(v,\lambda)\in\Req\mid v\notin\ActivationSet\,]$ \Comment{preserve the order of remaining labels}
    \State \Return $(\LS,\ILS,\Req)$
\EndProcedure
\Statex
\Statex Transforms the linear system to the one that would be obtained if the row $v$ was not present at the beginning.
\Procedure{RemoveInitialRow}{$\LS,\ILS,v,\activecount$}\label{proc:RemoveInitialRow}
    \State $\TrackingColumn \gets$ third-block column corresponding to row $v$
    \State $\AffectedRows \gets [\,r \mid (\LS)_{r,\TrackingColumn}=1\,]$
    \If{$\AffectedRows$ is empty}
        \State \Return $(\LS,\ILS)$
    \EndIf
    \State $\ReplacementRow\gets$ last row of $\AffectedRows$
    \For{$r\in \AffectedRows$ except $\ReplacementRow$}
        \State add row $\ReplacementRow$ of $\LS$ to row $r$ of $\LS$
    \EndFor
    \State add row $v$ of $\ILS$ to row $\ReplacementRow$ of $\LS$
    \State \Call{RestoreEchelonOrder}{$\LS,\ReplacementRow,\activecount$}
    \State \Return $(\LS,\ILS)$
\EndProcedure
\Statex
\Statex Activates columns of vertices in $X$. Ensures that the active part of the first block is in row echelon form.
\Procedure{ActivateVertices}{$\LS,\ILS,\Cols,\activecount,\ActivationSet$}\label{proc:ActivateVertices}
    \For{$v\in \ActivationSet$ in graph order} \Comment{activate one vertex completely before proceeding to the next}
        \If{$v\in\Cols$ and the column of $v$ is not within first $\activecount$ columns} \Comment{\ie if $v$ is non-active non-input}
            \State swap the first-block column of $v$ into position $\activecount+1$ in both $\LS$ and $\ILS$
            \State swap the matching entries in the column order $\Cols$
            \State $\activecount\gets \activecount+1$ \Comment{increase number of active columns}
            \State $\AffectedRows\gets[\,r\mid (\LS)_{r,\activecount}=1\,]$ \Comment{non-zero rows of the new active column}
            \If{$\AffectedRows$ is not empty}
                \State $\PivotRow\gets$ last row of $\AffectedRows$
                \For{$r\in \AffectedRows$ except $\PivotRow$}
                    \State add row $\PivotRow$ of $\LS$ to row $r$ of $\LS$ \Comment{Gaussian elimination by hand}
                \EndFor
                \State \Call{RestoreEchelonOrder}{$\LS,\PivotRow,\activecount$}
            \EndIf
        \EndIf
    \EndFor
    \State \Return $(\LS,\ILS,\Cols,\activecount)$
\EndProcedure
\Statex
\Statex Restores row echelon form via by-hand Gaussian elimination
\Procedure{RestoreEchelonOrder}{$\LS,\ChangedRow,\activecount$}\label{proc:RestoreEchelonOrder}
    \State use existing non-zero pivot rows, from top to bottom, to cancel pivot columns in row $\ChangedRow$
    \State find the pivot of row $\ChangedRow$ among the first $\activecount$ columns, if any
    \State move row $\ChangedRow$ to a suitable place among the non-zero pivot rows so the pivot order is increasing
    \State place all zero-coefficient rows after the non-zero pivot rows
\EndProcedure
\end{algorithmic}
\end{algorithm}

\begin{proof}
    We follow the pseudocode and prove the correctness of the operations appearing in Algorithms~\ref{algo:system construction} to~\ref{algo:row maintenance}.
    
    First, note that vertices in $I \cap O$ may be deleted from the \LOG at the beginning and then reinserted once the algorithm run is complete: This is because they are neither measured vertices nor can they appear in correction sets. Deleting them changes no entry in the reduced adjacency matrix. Thus, without loss of generality we assume $I \cap O = \emptyset$ from here onwards (\cf \cite[Theorem 3.10]{mitosekPauliFlowOpen2024}), as does the pseudocode.

    Fix any order of the vertices. We call this the \emph{graph order}. It is not the partial order of the gflow; instead, it is the order according to which columns are stored in various matrices throughout the computation.

    Let $\Rows = V \setminus O$ and $\Cols = V \setminus I$, both in graph order, be the lists of rows and columns, respectively, of the linear systems defined next.
    For every $v \in \Rows$ define a column vector:
    \begin{equation}
        \beta_{v,\XY} = e_v, \label{eq:xy-attached-vector}
    \end{equation}
    and, when $v\notin I$, also define two more column vectors:
    \begin{equation}
        \beta_{v,\YZ} = \Adj[\Rows,\,\{v\}]
        \qquad \text{and} \qquad
        \beta_{v,\XZ} = \Adj[\Rows,\,\{v\}]+e_v.\label{eq:yz-xz-attached-vectors}
    \end{equation}

    Let $\Req\sse \Rows\times\{\XY,\XZ,\YZ\}$ be the list of the pairs indexing these vectors and let $B = \{\beta_q\mid q\in \Req\}$ contain the corresponding columns. Thus $\Req$ contains only $(v,\XY)$ when $v\in I$, and all three pairs when $v\notin I$. This $\Req$ is part of the output of \Call{BuildLabellingSystem}{} in Algorithm~\ref{algo:system construction}. That method also returns two copies of the linear system with three blocks:
    \begin{equation}
        \ILS = \LS = \bigl[\,\Adj\left[\Rows,\ \Cols\right]\mid B\mid \Identity_{\Rows}\,\bigr] \label{eq:maintained-labelling-system}
    \end{equation}
    We refer to these blocks as the coefficient block, the attached block, and the maintenance block, similar to the approach used in \cite{mitosekAlgebraicInterpretationPauli2026, Mitosek2025thesis}. We never perform row operations on $\ILS$; however, unlike in \cite{mitosekAlgebraicInterpretationPauli2026, Mitosek2025thesis}, the system in $\LS$ sometimes undergoes column operations, which we mirror in the $\ILS$.

    Throughout, only some columns of the coefficient block are part of the linear system.
    These are the active columns of Definition~\ref{def:active}; we denote the set of active columns by $\Active$. The active columns will always be maintained to be leftmost in the coefficient block.

    Set the number of active columns $\activecount$ to $0$, initialise the set $\Solved$ of vertices to contain exactly the outputs $O$ and let the set of unsolved vertices be $\Unsolved = V \setminus \Solved = \Rows$. Call \Call{ActivateVertices}{} (line~\ref{proc:ActivateVertices} in the pseudocode) with activation set $O$. For each $v \in O$, this operation performs the following:\begin{enumerate}
        \item Swap column $v$ into position $\activecount + 1$ in the coefficient block of both $\LS$ and $\ILS$, and also in the list $\Cols$.
        \item Increase $\activecount$ by $1$.
        \item Performs cancellation of rows as follows: let $\PivotRow$ be the index of the last row of column $v$. Add row $\PivotRow$ to all other rows of $\LS$.
        \item Call \Call{RestoreEchelonOrder}{} (line~\ref{proc:RestoreEchelonOrder}) to ensure that the active part of the coefficient block remains in row echelon form. The procedure reduces the row using other pivot rows and then moves the reduced row to the position determined by its resulting pivot, or below all pivot rows if its coefficient part is identically zero. See \cite{mitosekAlgebraicInterpretationPauli2026, Mitosek2025thesis} for a more detailed description of an analogous approach in a slightly different context.
    \end{enumerate}
    Thus, at this point, all outputs are initial in the coefficient block, $\Active = O$, and $\activecount = \abs{O}$. The total cost of operations up to this point is $\bigO(n^3)$ (or, more precisely, it is $\bigO(n^2 \abs{O})$), because \Call{ActivateVertices}{} runs in $\bigO(n^2 |\ActivationSet|)$ where $\ActivationSet$ is the set of activated vertices.

    Now, we go into the main loop of the algorithm that is repeated until $\Solved = V$.
    The loop below maintains the following invariant. Let $\Unsolved = V \setminus \Solved$, then:
    \begin{itemize}
        \item The first $\activecount$ columns of the coefficient block correspond precisely to vertices in $\Solved \setminus I$, and their submatrix is in row-echelon form.
        \item Projecting $\LS$ onto its active columns of the coefficient block and the attached column $\beta_{v,\lambda}$ gives a system equivalent to:
        \begin{equation}
            \bigl[\, \Adj\left[\Unsolved,\ \Solved \setminus I\right] \mid \beta_{v,\lambda}[\Unsolved] \,\bigr] \label{eq:current-labelling-system}
        \end{equation} where $\beta_{v,\lambda}[\Unsolved]$ is the restriction of $\beta_v$ to entries in $\Unsolved$. This system is precisely the system from Lemma~\ref{lemma:linear systems}.
    \end{itemize}
    Note that, at the start, the loop invariant is satisfied.
    \begin{enumerate}
        \item Call \Call{FindNewSolutions}{} (line~\ref{proc:FindNewSolutions}), which detects vertices that can be solved with forward cones contained in $\Solved$: Let $Z$ be the set containing those rows in which the
        first $\activecount$ columns of $\LS$ are zero. For every
        $(v,\lambda)\in\Req$, inspect its attached column: the check of $(v,\lambda)$ leads to a consistent system precisely when this column is zero on every row in $Z$.
        Let $\Found$ contain every consistent request and let $\New$ be the list of the
        vertices appearing in $\Found$, in graph order. If $\Found$ is
        empty, return that no lgflow exists.
        
        Since the active part of the coefficient block is in row echelon form, we can detect vertices in $\New$ with a single read of the $\LS$ matrix, \ie in time $\bigO(n^2)$. Furthermore, thanks to the second part of the loop invariant, we can record the correction sets corresponding to each vertex in $\New$ by following Lemma~\ref{lemma:linear systems} in time $\bigO(n^2)$ per vertex. Since each vertex is solved at most three times (once per detected plane), the backtracking amortises to $\bigO(n^3)$ over the entire main loop.

        \item Discard the attached columns of newly-solved vertices from both $\LS$ and $\ILS$ by calling \Call{DiscardAttachedColumns}{} (line~\ref{proc:DiscardAttachedColumns}): the columns of the solved vertices are no longer necessary to maintain in the linear system\footnote{In \cite{mitosekAlgebraicInterpretationPauli2026, Mitosek2025thesis}, such columns were maintained for easier implementation and were skipped from consideration on later layers. Both approaches yield the same results, with the previous approach being easier to implement and the approach used here leading to a slight speed-up due to maintenance of smaller linear systems.}.

        \item Next, we bring the linear system to a form equivalent to the one that would be obtained if the vertices in $\New$ had never been present (or, more precisely, if all such vertices were outputs). 
        To do this, repeat for all $v \in \New$:\begin{enumerate}
            \item Remove the dependence of $\LS$ on the original row of $v$ by calling \Call{RemoveInitialRow}{} (line~\ref{proc:RemoveInitialRow}). The operation can be performed in $\bigO(n^2)$ time. Since each vertex appears in $\New$ at most once, this operation amortises to $\bigO(n^3)$ over the entire main loop. For the technical details, see \cite{mitosekAlgebraicInterpretationPauli2026, Mitosek2025thesis} or the pseudocode.
            The general idea of dependency removal is to first cancel all but one dependency by adding the last row with dependency to all earlier rows with dependency (where rows with dependency can be read off from the third block), then the original row of $v$ (accessible in $\ILS$) is added to the only remaining dependency in $\LS$, and lastly, by-hand Gaussian elimination is used to update the changed row so that the system maintains the row echelon form.
            
            \item Activate the column of $v$. This is again done by calling \Call{ActivateVertices}{} (line~\ref{proc:ActivateVertices}). Afterwards, the active part of the coefficient block contains $\Solved\cup\{v\}$; and the active part of the coefficient block is in row echelon form.
            
            \item Update $\Solved = \Solved \cup \New$.
        \end{enumerate}

        \item Lastly, update $\Unsolved = V \setminus \Solved$.
        After that, the active part of the coefficient block in $\LS$ is in row echelon form and consists of all columns corresponding to vertices in $\Solved \setminus I$, \ie the first part of the loop invariant is satisfied. Although there are $|\comp{O}|$ rows, they are all linear combinations of the rows present in $\Unsolved$. This means that the active part of the coefficient block is equivalent to the coefficient part of the linear system in Equation~\ref{eq:current-labelling-system}, \ie the second part of the loop invariant is also satisfied. \qedhere
    \end{enumerate}
\end{proof}

The algorithm solves each vertex $v$ as soon as it can measure $v$ in some plane.
This is closely related to the `maximally delayed' property of a gflow on an open graph with a fixed labelling; see \cite{backensThereBackAgain2021,mhallaFindingOptimalFlows2008a,simmonsRelatingMeasurementPatterns2021}. 
Indeed, the returned flow $(c,\prec)$ for the \LOG $\lGIO$ is maximally delayed.

At the same time, the flow returned by the above algorithm is not guaranteed to be focused (\cf Definition~\ref{def:focusing}).
If a focused flow is necessary, then one can focus the found flow via a stepwise process of modifying the correction sets \cite[Section~3.3]{backensThereBackAgain2021} or run a $\bigO(n^3)$ flow-finding procedure on $\lGIO$, which is guaranteed to find a focused flow \cite{mitosekAlgebraicInterpretationPauli2026}.

The algorithm can be adjusted to search for measurement labellings that support additional requirements, such as restricting the measurement labels allowed for each vertex, by adjusting the filtering inside $\Call{CandidateLabels}{}$. In particular, if each vertex is restricted to only one allowed label, the problem collapses to standard gflow finding.
One may also restrict how the constructed partial order looks by filtering which vertices can be considered at each layer. In particular, one could supply a strict partial order and look only for measurement labellings compatible with it. Thus, our approach not only allows finding the extension of $(G,I,O)$ to $(G,I,O,\ld,c,\prec)$ where $(c,\prec)$ is a gflow on $(G,I,O,\ld)$, but also extending $(G,I,O,\ld)$, $(G,I,O,\prec)$, or any setting with $\ld$ and/or $\prec$ partially supplied.

\section{Properties of balancing}
\label{sec:balancing}

The algorithm in Section~\ref{sec:algorithm} finds an lgflow if one exists.
Given a fixed labelling, it is well-known that there is at most one focused gflow for a balanced \LOG \cite{mhallaWhichGraphStates2014a,simmonsRelatingMeasurementPatterns2021,mitosekAlgebraicInterpretationPauli2026} but there may be many focused gflows for unbalanced ones.
Here, we consider a similar question for lgflow: in particular, we study the \lspace\ resulting in gflow.
Again, different properties emerge depending on whether the open graph is balanced.
In particular, we show that for a balanced open graph there is at most one measurement labelling resulting in gflow.
In other words, the structure of the balanced open graph alone implies which measurement planes must be chosen for computation to be deterministic.
The same is not true for unbalanced open graphs.
The uniqueness property for balanced graphs is often useful when working with focused gflow, for example in the context of flow-preserving transformations: the uniqueness property for lgflow will likely have similar advantages.
We propose a balancing mechanism that transforms an unbalanced \LOG into a balanced one while making minimal changes to the measurement labelling by identifying vertices that can be made into inputs.

We also provide a probabilistic version of the lgflow-finding algorithm that can reach any valid measurement labelling resulting in gflow while maintaining $\bigO(n^3)$ complexity.
This enables the full \lspace\ to be analysed for any given open graph, with applications to the question of which computations are supported by some given open graph state that can be prepared in a lab, or in the study of foundational aspects of deterministic MBQC.

\subsection{Uniqueness for balanced open graphs}

We now show that there exists at most one labelling resulting in gflow for balanced instances of $\mainproblem$, \ie those with $\abs{I}=\abs{O}$.
This means that not only is focused gflow unique on balanced \LOG{}s (if it exists), but focused lgflow is also unique on balanced open graphs (again if it exists).

\begin{theorem}[Uniqueness]\label{Th:uniqueness}
    Let $(G,I,O)$ be a balanced open graph.
    Then there exists at most one measurement labelling $\ld$ such that $(G,I,O,\ld)$ has gflow.
\end{theorem}
\begin{proof}
    Suppose for contradiction that the hypothesis does not hold, and let $(G,I,O)$ be the open graph with the fewest non-outputs for which the hypothesis does not hold.
    Therefore, there exists $\ld \ne \ld'$ such that both $(G,I,O,\ld)$ and $(G,I,O,\ld')$ have gflows.
    Let $v \in \comp{O}$ be such that $\ld(v) \ne \ld'(v)$.
    Then necessarily $v \notin I$, as at least one of $\ld(v)$ or $\ld'(v)$ must be different from $\XY$, and inputs in labelled open graphs with gflow cannot be measured in a basis other than $\XY$.
    We consider three cases depending on $\{ \ld(v), \ld'(v) \}$.
    \begin{enumerate}
        \item Suppose $\{ \ld(v), \ld'(v) \} = \{ \YZ, \XZ \}$.
        Then by Lemma~\ref{lemma:Z-like removal} $v$ can be removed from both $(G,I,O,\ld)$ and $(G,I,O,\ld')$ while preserving the existence of gflows, resulting for both cases in the same smaller open graph $(G-v,I,O)$.
        Thus, either we get a smaller counterexample in the number of non-outputs, or $v$ is the only vertex on which $\ld$ and $\ld'$ differ.
        The first case contradicts the minimality assumption.
        
        In the second case, let $\hat{\ld}$ be the restriction of $\ld$ (and also of $\ld'$) to $\comp{O} \setminus \{ v \}$.
        Then, $(G-v,I,O,\hat{\ld})$ has gflow.
        Moreover, the insertion of a new vertex $v$ with neighbours $N_G(v)$ will lead to the labelled open graphs $(G,I,O,\ld)$ if the new vertex is measured $\YZ$ and to $(G,I,O,\ld')$ if the new vertex is measured $\XZ$.
        Both of the resulting labelled open graphs have gflow by assumption.
        However, when inserting a vertex with a fixed neighbourhood into a balanced open graph, at most one of the planes $\YZ$ and $\XZ$ can lead to gflow by Observation~\ref{obs:Z-like insertion}, so again we get a contradiction.
        
        \item Suppose $\{ \ld(v), \ld'(v) \} = \{ \XY, \YZ \}$.
        Since $v \notin I$, by Theorem~\ref{th:LC}, we may perform a local complementation about $v$ in a gflow-preserving way.
        After this local complementation, $\{ \ld(v), \ld'(v) \} = \{ \XZ, \YZ \}$, \ie we land in the previous case, which leads to a contradiction.
        
        \item Finally, suppose $\{ \ld(v), \ld'(v) \} = \{ \XY, \XZ \}$.
        By the existence of a gflow, $\comp{I} \cap N(v) \ne \emptyset$, since at least one neighbour of $v$ must be in the correction set of $v$.
        Let $u \in \comp{I} \cap N(v)$; then, by Theorem~\ref{th:LC}, we may perform a local complementation about $u$ in a gflow-preserving way.
        After this local complementation, $\{ \ld(v), \ld'(v) \} = \{ \XY, \YZ \}$, \ie we land in the previous case, which leads to a contradiction.
    \end{enumerate}
    Thus, each case reduces to the first, which leads to a contradiction; therefore, the hypothesis must hold.
\end{proof}

For example, the open graph in Figure~\ref{fig:uniqueness-example} has $\abs{I}=\abs{O}$, meaning it is balanced, and therefore the given measurement labelling is the unique labelling resulting in gflow.

\subsection{Triple corrigibility}

When a labelled open graph is unbalanced, the above uniqueness result does not apply: many different measurement labellings may result in gflow.
It actually turns out that vertices on which measurement labellings can differ have special properties.
Firstly, whenever there exist two lgflows whose labellings differ only on a single vertex, then this vertex can be measured in all three planes.

\begin{theorem}\label{Th:if two planes then three planes}
    Let $(G,I,O)$ be an open graph and let $\ld$, $\ld'$, and $\ld''$ be three different measurement labellings that differ only on a vertex $a$.
    Suppose that $(G,I,O,\ld)$ and $(G,I,O,\ld')$ have gflow.
    Then $(G,I,O,\ld'')$ also has gflow.
\end{theorem}

\begin{proof}
    Fix a deterministic gflow-finding algorithm that works backwards from the outputs in layers to find a maximally delayed gflow.
    Let $(c,\prec)$ be the gflow this algorithm finds on $(G,I,O,\ld)$ and let $(c',\prec')$ be the gflow it finds on $(G,I,O,\ld')$.
    Define $d$ to be the depth of $a$ in $(c,\prec)$, \ie the longest path in the partial order from $a$ to an output, and define $d'$ to be the depth of $a$ in $(c',\prec')$; without loss of generality, assume $d \geq d'$ (otherwise swap the roles of the primed and the unprimed case).

    Note that for vertices with depth $< d'$, there can be no effect of the change in measurement label, as the flow-finding algorithm works backwards from the outputs.
    Thus, considering the primed correction function and the unprimed partial order, for any $u\in \left(c'(a) \cup \odd{c'(a)}\right) \setminus \{ a \}$, we have $\neg (u \prec a)$.
    Define a tuple $(c'',\prec'')$, where
    \begin{equation}\label{eq:def-cprimeprime}
        c''(v) := \begin{cases}
            c(a)\symd c'(a) &\text{if } v = a \\
            c(v) &\text{otherwise,}
        \end{cases}
    \end{equation}
    and $\prec''$ is the transitive closure of ${\prec} \cup \left(\{(a,w)\mid w\in c'(a) \vee w\in\odd{c'(a)}\}\setminus\{(a,a)\}\right)$.
    As we only introduce successors for $a$ and these successors are either already successors or previously incomparable to $a$, the relation $\prec''$ is a valid strict partial order.
    It is straightforward to check that $c''(a)$ is an appropriate correction set for $a$ with label $\ld''$.
    Hence $(c'',\prec'')$ is a gflow on $(G,I,O,\ld'')$.
\end{proof}

This leads to the following definition.

\begin{definition}\label{def:triple-correctable}
    Let $(G,I,O)$ be an open graph. We call a vertex $a \in \comp{O}$ \emph{triple-correctable} when there exist $\ld, \ld', \ld''$ pairwise differing only on $a$ with $(G,I,O,\mu)$ having gflow for all $\mu \in \{ \ld, \ld', \ld'' \}$.
\end{definition}

By Theorem~\ref{Th:if two planes then three planes}, any vertex that can be measured in two planes can be measured in all three planes, \ie is triple-correctable.
Note that it is important for triple-corrigibility that the measurement labels of other vertices do not change across the different labellings.
In other words, triple-corrigibility is relative to the measurement labels of the other vertices, as demonstrated in the following example.

\begin{example}
    Consider the path graph of length 3, \ie $G=(V,E)$ with $V = \{a,b,c\}$ and $E = \{ \{a,b\}, \{b,c\}\}$.
	Suppose $I = \emptyset$ and $O=\{b\}$.
    \ctikzfig{LOGs/triple-correctable-counterex}
    Then $\ld(a)=\YZ$ and $\ld(c)=\XY$ is a labelling that leads to gflow: we may take $g(a) = \{a\}$ and $g(c) = \{b\}$, with $c\prec a\prec b$.
    In fact, with $\ld(a)=\YZ$, the vertex $c$ is triple-correctable: we could instead change its measurement label to $\YZ$ and its correction set to $\{c\}$ or we could change its measurement label to $\XZ$ and its correction set to $\{b,c\}$.
    
    On the other hand, while $\ld(c)$ is fixed to $\XY$, changing the measurement label of vertex $a$ would break the gflow: the labels $\XY$ and $\XZ$ both require a neighbour in their correction set.
    Yet if $b$ is in both correction sets, then $a\in\odd{g(c)}$ and $c\in\odd{g(a)}$, which means there is no consistent partial order.
    Thus, in the context of the given labelling, $a$ is not triple-correctable.
    This difference in terms of triple-corrigibility arises even though in the unlabelled graph the roles of $a$ and $c$ are entirely symmetric.
\end{example}

Having several flows on \LOG{}s that differ only in one measurement label gives more flexibility in focusing.
In particular, we can ensure that the vertex $a$ (that can have different measurement labels) does not appear in correction sets of other vertices, even in the context of the labelling where $a$ is $\XY$-measured.

\begin{corollary}\label{cor:not-in-correction-sets}
    Let $(G,I,O)$ be an open graph and let $\ld$, $\ld'$, and $\ld''$ be three different measurement labellings that differ only on the vertex $a\in\comp{O}$.
    Suppose that $(G,I,O,\mu)$ has gflow for all $\mu\in\{\ld,\ld',\ld''\}$.
    Then for any such $\mu$ there exists a focused gflow on $(G,I,O,\mu)$ such that $a$ does not appear in correction sets (other than potentially its own).
\end{corollary}
\begin{proof}
    The result follows from basic focusing if $\mu(a)\in\{\YZ,\XZ\}$, so we only need to consider the case where $\mu(a)=\XY$.
    In the proof of Theorem~\ref{Th:if two planes then three planes}, note that the correction function $c''$ defined in \eqref{eq:def-cprimeprime} for $(G,I,O,\ld'')$ differs from $c$ only on $a$.
    We could similarly define an alternative correction function $c'''$ for $(G,I,O,\ld')$ by
    \[
        c'''(v) := \begin{cases}
            c'(a) &\text{if } v = a \\
            c(v) &\text{otherwise,}
        \end{cases}
    \]
    with an analogously defined partial order, to arrive at a second valid gflow for $(G,I,O,\ld')$.
    Thus there exist (not necessarily focused) gflows on all three labelled open graphs whose correction functions differ only for $a$.
    Denote these three correction functions by $c_\XY, c_\XZ, c_\YZ$ depending on the respective label of $a$ in the \LOG\ to which they apply.
    Let $\prec$ be the minimal partial order such that $(c_{\XY},\prec)$ is a gflow.
    Define
    \[
        \hat{c}(v) := \begin{cases}
            c_{\XY}(v) \symd c_{\YZ}(a) &\text{if } a\in c_{\XY}(v) \\
            c_{\XY}(v) &\text{otherwise},
        \end{cases}
    \]
    then this has the desired property of $a$ not appearing in correction sets.
    It still satisfies conditions \conref{g3}--\conref{g5} because $c_\XY(v) = c_\YZ(v)$, so $v\notin c_\YZ(a)\cup \odd(c_\YZ(a))$ as this would contradict the partial order conditions for the gflow induced by $c_\YZ$.

    Next define $\hat{\prec}$ to be the transitive closure of ${\prec} \cup \{(v,w)\mid a\in c_{\XY}(v) \wedge w\in c_{\YZ}(a) \cup \odd{c_{\YZ}(a)}\}$.
    The pair $(\hat{c}, \hat{\prec})$ is a gflow if and only if $\hat{\prec}$ is a partial order, which remains to prove.
    In the following, we will write $w\preceq v$ as a shorthand for\footnote{This is a different convention than that used in some related work \cite[Definition~4.1]{simmonsRelatingMeasurementPatterns2021}, where $u \preceq v$ means $\neg(v \prec u)$.} $w = v \vee w\prec v$.
    
    Suppose for a contradiction that $\hat{\prec}$ is not a strict partial order, then there exists $v$ such that $a\in c_{\XY}(v)$ and $w\in c_{\YZ}(a) \cup \odd{c_{\YZ}(a)}$ such that $w\preceq v$ holds.
    As $a\in c_{\XY}(v)$, this then implies that $w\preceq v\prec a$.
    But $w\prec a$ implies there is a sequence $w_0,\ldots,w_k$ such that $w=w_0$ and $w_k=a$ and for each $j\in [k-1]$, we have $w_{j+1}\in c_{\XY}(w_j)\cup\odd{c_{\XY}(w_j)}$.
    Moreover, $w_0,\ldots, w_{k-1} \neq a$ and thus $c_{\XY}(w_j)=c_{\YZ}(w_j)$ for all $j\in [k-1]$.
    But then $c_{\YZ}$ cannot be a valid gflow as $w$ must both precede and succeed $a$ in any associated partial order, a contradiction.
    Hence $\hat{\prec}$ is indeed a strict partial order and $(\hat{c},\hat{\prec})$ is a valid gflow.
    
    This flow can be focused by standard techniques and the focusing process will not re-introduce $a$ to the correction sets of any vertices in $\comp{O}\setminus\{a\}$.
\end{proof}

Having a gflow in which an $\XY$-measured vertex does not appear in correction sets means that this vertex can be turned into an input:

\begin{theorem}\label{Th:if two planes then input}
    Let $(G,I,O)$ be an open graph and let $\ld$ and $\ld'$ be two different measurement labellings that differ only on $a$.
    Suppose that $(G,I,O,\ld)$ and $(G,I,O,\ld')$ have gflow.
    Then $(G,I \cup \{ a \},O,\ld_{\XY})$ also has gflow where $\ld_{\XY}(a) = \XY$ and $\ld_{\XY}(u) = \ld(u) = \ld'(u)$ for $u \ne a$.
\end{theorem}

\begin{proof}
    By Theorem~\ref{Th:if two planes then three planes}, the \LOG $(G,I,O,\ld_{\XY})$ has gflow since the assumption of having two labellings that result in gflow implies that in fact all three possible labellings result in gflow.
    Moreover, by Corollary~\ref{cor:not-in-correction-sets}, $(G,I,O,\ld_{\XY})$ has a gflow in which $a$ does not appear in correction sets; denote this gflow by $(c,\prec)$.
    Then $(c,\prec)$ is also a gflow on the \LOG with $a$ turned into an input: $(G,I\cup\{a\},O,\ld_{\XY})$, ending the proof.
\end{proof}

Therefore, any triple-correctable vertex can be turned into an input.

\begin{corollary}\label{cor:triple-correctable then input}
    Let $(G,I,O)$ be an open graph with lgflow and let $v \in V$ be triple-correctable. Then $(G,I\cup\{v\},O)$ also has lgflow.
\end{corollary}

\begin{proof}
    Immediate by combining Definition~\ref{def:triple-correctable} with Theorem~\ref{Th:if two planes then input}.
\end{proof}

Running the algorithm from Section~\ref{sec:algorithm} on an unbalanced open graph need not detect any triple-correctable vertices, thus giving no candidate for a vertex that may be turned into an input.
This is because assigning different measurement labels to a vertex (while keeping the rest of the labelling the same) may mean that it appears at different depth in the resulting gflow, while our algorithm only identifies the minimal-depth labellings.
For example, in Appendix~\ref{subsec:worked example unbalanced}, the starting open graph is unbalanced, but the algorithm does not flag any vertex as triple-correctable.
Yet, as we show in the next subsection, we can still efficiently detect vertices that can be turned into inputs to transform an unbalanced open graph into a balanced one.

\subsection{Balancing the number of inputs and outputs}\label{subsec:balancing}

So far, we have considered various aspects of flow on open graphs, dropping the requirement for the labelling $\ld$ to be specified in advance.
We now go one step further and show that, for an unbalanced \LOG $(G,I,O,\ld)$ (\ie with $\abs{I}<\abs{O}$) that has gflow, additional inputs compatible with the continued existence of gflow can also be identified efficiently.
By an `additional input', we here mean a vertex $a\in\comp{I}$ such that $(G,I\cup\{a\},O,\ld')$ has gflow, where $\ld'(a)=\XY$ (unless $a$ is an output) and $\ld'(v)=\ld(v)$ for all $v\in\comp{O}\setminus\{a\}$.

We already have one way of identifying some vertices that can be made into inputs: by Theorem~\ref{Th:if two planes then input}, if there exist at least two measurement labellings for the open graph $(G,I,O)$ that are compatible with gflow and that differ only in the label assigned to a single vertex $v$, then $(G,I\cup\{v\},O)$ also has lgflow.
This condition is sufficient but not necessary; for example, it cannot identify outputs that can also become inputs.

Instead, we will make use of the focused sets of Definition~\ref{def:focusing} that form a vector space of dimension $\abs{O}-\abs{I}$ and thus are intuitively linked to how `unbalanced' the open graph is.
First, define the `extended focused set' to be the set of vertices that receive a non-trivial effect from the graph stabiliser that is implicitly defined by this focused set.

\begin{definition}\label{def:extended-focused-set}
	Let $(G,I,O,\ld)$ be an open graph and let $\focusedset$ be one of its focused sets.
	The corresponding \emph{extended focused set} is $\extfocset = \focusedset \cup (\odd{\focusedset}\setminus O)$.
	Define $\allfocused$ to be the set of all focused sets, and let $\extallfoc := \bigcup_{\focusedset\in\allfocused} \extfocset$.
\end{definition}

There are technical reasons for excluding from the definition of an extended focused set those outputs that are simply in the odd neighbourhood of the original focused set, \cf Remark~\ref{rem:extended-focused-set} at the end of this section.

While the number of distinct focused sets grows exponentially with the difference $\abs{O}-\abs{I}$, to determine $\extallfoc$, it suffices to know a spanning set of the vector space $\allfocused$: $v\in\extallfoc$ if and only if there exists an element $\focusedset$ of the spanning set such that $v\in\extfocset$.
This spanning set need not be a basis, but $\abs{O}-\abs{I}$ elements suffice to characterise the entire space.

Next, we show that the set $\extallfoc$ defined above does not contain any inputs, so we will not need to worry about candidate vertices from this set already being inputs.

\begin{lemma}\label{lem:focused-set-inputs-outputs}
	Let $\Gamma = (G,I,O,\ld)$ be a labelled open graph with gflow, then $I\cap\extallfoc =\emptyset$.
\end{lemma}	
\begin{proof}
	Let $\focusedset$ be any focused set of $\Gamma$.
	By Definition~\ref{def:focusing}, $\focusedset\sse\comp{I}$, \ie inputs do not appear in focused sets.
	Moreover, by Remark~\ref{rem:inputs-XY}, all inputs are either outputs or $\XY$-measured.
	Thus, condition~\conref{FZ} means that \emph{measured} vertices in $\odd{\focusedset}$ cannot be inputs: hence $I\cap\odd{\focusedset} \sse O$.
	But outputs in the odd neighbourhood were excluded from the definition of the extended focused set, so $I\cap\extfocset = \emptyset$.
	As $\focusedset$ was arbitrary, we have $I\cap\extallfoc = \bigcup_{\focusedset\in\allfocused} (I\cap \extfocset) = \emptyset$.
\end{proof}

\begin{theorem}\label{thm:add-input}
	Let $\Gamma = (G,I,O,\ld)$ be a labelled open graph which satisfies $\abs{I}<\abs{O}$ and has focused gflow $(g,\prec)$.
	Let $\extallfoc$ be the set of all vertices that appear in some extended focused set of $\Gamma$.
	Suppose $a\in\extallfoc$ is minimal with respect to $\prec$ among vertices in $\extallfoc$.
	Define $\ld':\comp{O}\to\{\XY,\XZ,\YZ\}$ as follows:
	\[
		\ld'(v) := \begin{cases}
				\XY &\text{if } v = a \\
				\ld(v) &\text{otherwise.}
			\end{cases}
	\]
	Then $(G,I\cup\{a\},O,\ld')$ has gflow.
\end{theorem}
\begin{proof}
	Let $\focusedset$ be a focused set of $\Gamma$ that satisfies $a\in\extfocset$.
	We will distinguish two cases: either $a$ is `$X$-like', meaning it is an output or $\XY$-measured, or $a$ is `$Z$-like', meaning it is $\XZ$- or $\YZ$-measured.
	\begin{itemize}
		\item Suppose either $a\in O$, or $a\notin O$ and $\ld(a)=\XY$, then the overall measurement labelling does not change.
		Hence the only potential obstacle to having a gflow on $(G,I\cup\{a\},O,\ld)$ is if $a$ appears in correction sets of $g$.
		Now, by Definitions~\ref{def:focusing} and~\ref{def:extended-focused-set}, $a\in\extfocset$ implies $a\in\focusedset$.
		Define
		\[
			g'(v) := \begin{cases}
					g(v)\symd\focusedset &\text{if } a\in g(v) \\
					g(v) &\text{otherwise,} \\
				\end{cases}
		\]
		then $a\notin g'(v)$ for all $v\in\comp{O}$.
		Conditions \conref{g3}--\conref{g5} are satisfied by $g'$ because $a\in g(v)$ implies $v\prec a$ whereas all measured elements of $\focusedset\cup\odd{\focusedset}$ are incomparable to or greater than $a$: therefore, $v\in g'(v) \Leftrightarrow v\in g(v)$ and $v\in\odd{g'(v)} \Leftrightarrow v\in\odd{g(v)}$.
		Define $\prec'$ to be the transitive closure of ${\prec} \cup (\{v\in\comp{O} \mid a\in g(v) \}\times(\focusedset\cup\odd{\focusedset}))$, then \conref{g1} and \conref{g2} hold for all vertices.
		It remains to show that $\prec'$ is a strict partial order.
		Now, the only way for $\prec'$ to contain a cycle is if there exists some $v$ such that $a\in g(v)$ and some $w\in\focusedset$ such that $w\prec v$.
		But $a\in g(v)$ implies $v\prec a$ by \conref{g1}.
		Therefore $w\prec a$ by transitivity, contradicting the assumption that $a$ is minimal with respect to $\prec$ within $\extallfoc$.
		Hence $\prec'$ is a strict partial order and $(g',\prec')$ is a gflow on $(G,I\cup\{a\},O,\ld')$.
		
		\item Suppose $a\notin O$, and $\ld(a)\in\{\XZ,\YZ\}$, then $a\notin g(v)$ for any $v\in\comp{O}\setminus\{a\}$ by \conref{FX}.
		Thus the only potential obstacle to having a gflow on $(G,I\cup\{a\},O,\ld')$ is the correction set of $a$ itself, which needs to change because the measurement label of $a$ is changing.
		Now, by Definitions~\ref{def:focusing} and~\ref{def:extended-focused-set}, $a\in\extfocset$ implies $a\in\odd{\focusedset}$ because $a\notin\focusedset$ by the definition of focusing.
		Define
		\[
			g'(v) := \begin{cases}
				\focusedset &\text{if } v = a \\
				g(v) &\text{otherwise,} \\
			\end{cases}
		\]
		and let $\prec'$ to be the transitive closure of ${\prec} \cup \{(a,w)\mid w\in \focusedset\cup\odd{\focusedset}\}$.
		Then $(g',\prec')$ satisfies the gflow conditions for all vertices (including $a$), and it only remains to show that $\prec'$ is a strict partial order.
		That argument is analogous to the previous case.
	\end{itemize}
	We have constructed an explicit gflow on $(G,I\cup\{a\},O,\ld')$ for any possible type of $a$, as desired.
	Note that the gflows constructed in either case are again focused.
\end{proof}

This suggests the following algorithm for equalising the number of inputs and outputs given an unbalanced \LOG.

\begin{theorem}\label{thm:balancing algo}
    There exists an $\bigO(n^3)$ algorithm that, given an unbalanced open graph $(G,I,O,\ld)$ with $n = |V|$, finds $J \subseteq \left(V \setminus I\right)$ such that $(G,I\cup J,O,\ld')$ has gflow, where $\ld'(v) = \XY$ for $v \in J$ and $\ld'(v) = \ld(v)$ otherwise.
\end{theorem}
\begin{proof}
Building on Theorem~\ref{thm:add-input}, the algorithm is immediate:
\begin{enumerate}
	\item Find a focused gflow $(g,\prec)$ as well as a basis $\{\focusedset_k\}_{k=1}^{\abs{O}-\abs{I}}$ for the space of focused sets on $\Gamma$.
	\item Pick a vertex $a$ that is minimal according to $\prec$ in $\extallfoc = \bigcup_{k=1}^{\abs{O}-\abs{I}} \extfocset_k$.
	Theorem~\ref{thm:add-input} describes how to construct a gflow $(g',\prec')$ for $(G,I\cup\{a\},O,\ld')$ and this gflow is again focused.
	\item If the new \LOG still has more outputs than inputs, update the basis for the space of focused sets as follows:
    \begin{itemize}
        \item Choose one focused set $\focusedset_j$ in the basis such that $a\in\extfocset_j$, and remove it from the basis.
        \item For every other focused set $\focusedset_\ell$ in the basis that satisfies $a\in\extfocset_\ell$, replace $\focusedset_\ell$ by $\focusedset_\ell\symd\focusedset_j$.
        \item Then go back to step 2 with the updated \LOG, updated gflow, and updated basis.
    \end{itemize} 
	Otherwise, stop.
\end{enumerate}
We now show that the algorithm runs in $\bigO(n^3)$ time.
To see this, we can upper-bound the time for each step.\begin{enumerate}
    \item Finding a basis for the space of focused sets of a fixed \LOG takes $\bigO(n^3)$; recall that this space corresponds precisely to the kernel of the flow-demand matrix \cite[Theorem~3.23]{mitosekAlgebraicInterpretationPauli2026}.
    Once the elements $\{ \focusedset_k \}$ making up this basis are known, we can compute the set $\{ \extfocset_k \}$ and maintain both from this point.
    At this point, the presentation of $\prec$ might be a $\trl_c$ matrix if we follow the algorithm in \cite{mitosekAlgebraicInterpretationPauli2026}.
    By using the Floyd–Warshall algorithm, we can compute the transitive closure of $\trl_c$, \ie the matrix of $\prec$, in $\bigO(n^3)$ \cite[Section~23.2]{cormenIntroductionAlgorithmsFourth2022}.
    \item To pick a minimal $a$, we first compute the union of the sets $\{ \extfocset_k \}$ (taking at most time equivalent to the total size of these sets, which is definitely within $\bigO(n^2)$), and then check for the minimal element according to $\prec$ within this union.
    Since $\prec$ is at this point given as a matrix, comparing two elements takes $\bigO(1)$, and thus the minimal element $a$ can be found in $\bigO(n)$, \ie time proportional to the size of the union set. In total, the second step therefore takes $\bigO(n^2)$.
    \item The subparts of the third step are completable in $\bigO(n^2)$ again:\begin{itemize}
        \item Finding and removing the first set within $\{ \extfocset_k \}$ containing $a$ can be done by looping through all sets, so at most $\bigO(n^2)$ time.
        \item Updating all sets can again be done in $\bigO(n^2)$ by looping through all sets and applying symmetric difference to those containing $a$.
    \end{itemize}
    Thus, the third step again takes $\bigO(n^2)$ in total.
    
    Now, the second and third step are each repeated a total of $|O|-|I|$ times, which is within $\bigO(n)$.
    Thus, the total time spent in both steps is $\bigO(n^3)$, which combined with the $\bigO(n^3)$ time for the first step results in $\bigO(n^3)$ overall\footnote{In fact, if $(c,\prec)$ is supplied up front, a slightly better bound of $\tilde{\bigO}(n^{\omega} + (|O|-|I|)^2 n)$ can be obtained by methods based on matrix multiplication for the first two steps \cite{StrassenVolkerMatmul,FischerMeyerTransitive,burgisserComputationalProblemsLinear1991} and observing that initially the size of the basis of space of focused sets is $|O|-|I|$ and the total number of repetitions for the last two steps is also $|O|-|I|$ which may be smaller than $\bigO(n)$. Here, $\omega$ is the exponent of matrix multiplication, for instance $\omega \le \log_2 7 \approx 2.81$ via Strassen's algorithm \cite{strassenGaussianEliminationNot1969}. Such methods are not practical unless $n$ is sufficiently large (for example, see \cite{huangStrassenAlgorithmReloaded2016}); even better bounds are `galactic' and never suitable in practice, for instance see \cite{dupont2026improvingmatrixmultiplicationexponent}.}. \qedhere
\end{enumerate}
\end{proof}

\begin{remark}\label{rem:extended-focused-set}
	We can now see why outputs that appear only in $\odd{\focusedset}$ have been excluded from the definition of $\extfocset$.
	Consider the path graph of length 3, \ie $G=(V,E)$ with $V = \{a,b,c\}$ and $E = \{ \{a,b\}, \{b,c\}\}$.
	Suppose $I = \{a\}$, $O=\{b,c\}$, and $\ld(a)=\XY$.\begin{equation*}
	    \begin{tikzpicture}
	\begin{pgfonlayer}{nodelayer}
		\node [style=GR Tiny Empty] (1) at (0, 0) {};
		\node [style=GR Tiny Empty] (2) at (1.5, 0) {};
        \node [style=ZH H] (3) at (-1.5, 0) {};
		\node [style=GR Tiny Black] (4) at (-1.5, 0) {};
		\node [style=BlackTEXT] (5) at (-1.5, -0.5) {$\scriptstyle a$};
		\node [style=BlackTEXT] (6) at (0, -0.5) {$\scriptstyle b$};
		\node [style=BlackTEXT] (7) at (1.5, -0.5) {$\scriptstyle c$};
	\end{pgfonlayer}
	\begin{pgfonlayer}{edgelayer}
		\draw (4) to (1);
        \draw (1) to (2);
	\end{pgfonlayer}
\end{tikzpicture}

	\end{equation*}
	The \LOG $(G,I,O,\ld)$ has a gflow given by $g(a)=\{b\}$ and $a\prec b,c$.
    It has a unique focused set $\focusedset = \{c\}$; the corresponding extended focused set is $\extfocset = \focusedset = \{c\}$.
	It is straightforward to see that $c$ can be made into an input without losing the existence of gflow.
	Yet $b$ cannot be made into an input, even though $b\in\odd{\focusedset}$, because it is needed to correct $a$.
	This is why Definition~\ref{def:extended-focused-set} excludes outputs that appear only in $\odd{\focusedset}$ from $\extfocset$.
    
    The issue does not arise for $\YZ$- or $\XZ$-measured vertices because, in a focused gflow, these do not appear in correction sets other than their own.
\end{remark}

The space of focused sets of a labelled open graph is independent of the choice of gflow.
Yet choosing a different focused gflow when applying Theorem~\ref{thm:add-input} may lead to different sets of inputs.
This is because different partial orders will affect which vertices are minimal in $\extallfoc$.
In the above proof, we considered only local changes to the measurement labelling, in the sense that the measurement labelling remains invariant except that if a $\YZ$- or $\XZ$-measured vertex that becomes an input, its label changes.
If the measurement labelling is allowed to be changed globally, this may lead to even more candidates for vertices that can become inputs; see Appendix~\ref{subsec:worked example balancing} for an example.

\subsection{Exploration of the \lspace}

The methods in Section~\ref{sec:algorithm} and Subsection~\ref{subsec:balancing} can find measurement labellings that result in gflow and additional inputs that balance any labelled open graph, both while preserving the existence of gflow.
These methods efficiently find a valid solution but are not suited for enumerating solutions.
Here, we present a modified probabilistic version of the algorithm from Section~\ref{sec:algorithm} that is less efficient but has a non-zero probability of finding any valid measurement labelling resulting in gflow.
By exploring the algorithm's entire decision tree, we can also enumerate all possible solutions to the measurement-labelling problem.

Before presenting the modified algorithm, we strengthen the relationship between potential inputs and triple-correctable vertices.
In Theorem~\ref{cor:triple-correctable then input}, we showed that triple-correctable vertices can become inputs.
We now show the reverse statement.

\begin{theorem}\label{th:if input then triple-correctable}
    Let $(G,I,O,\ld)$ be a \LOG with gflow and $a \in I \setminus O$ be a non-output input. Then $a$ is triple-correctable in $(G,I\setminus\{a\},O)$.
\end{theorem}

\begin{proof}
    Let $(c,\prec_c)$ be a focused gflow on $(G,I,O,\ld)$. By Observation~\ref{obs:inputs-minimal-outputs-maximal}, $a$ is minimal in $\prec_c$. Since $a \in I \setminus O$, by Remark~\ref{rem:inputs-XY} we have $\ld(a) = \XY$. Then $(c,\prec_c)$ is also a focused gflow  on $(G,I \setminus \{ a \},O,\ld)$. Define $\ld'$ as:\begin{gather*}
        \ld'(u) = \begin{cases}
            \YZ &\text{if } u = a \\
            \ld(u) &\text{otherwise}
        \end{cases}
    \end{gather*}
    We show that $(G,I\setminus\{a\},O,\ld')$ has gflow.
    Let $c'$ be a correction function for $(G,I\setminus\{a\},O,\ld')$ defined as:
    \begin{gather*}
        c'(u) = \begin{cases}
            \{a\} &\text{if } u = a \\
            c(a) &\text{otherwise}
        \end{cases}
    \end{gather*}
    Let $\prec$ be the transitive closure of ${\prec_c} \cup (\{ a \} \times 
    \left(V \setminus \{ a \}\right))$.
    Then $(c',\prec)$ is the desired gflow on $(G,I\setminus\{a\},O,\ld')$: since $a$ is minimal in $\prec_c$ and $(c,\prec_c)$ is a gflow, $a$ does not appear in correction sets of other vertices or their odd neighbourhoods.
    Moreover, $c'(a)$ satisfies \ref{g4} in $(G,I\setminus\{a\},O,\ld')$, meaning $a$ is appropriately corrected. Thus, $a$ can be measured in the $\XY$ and $\YZ$ planes and by Theorem~\ref{Th:if two planes then three planes}, $a$ is triple-correctable, ending the proof.
\end{proof}

Therefore, triple-correctable vertices and inputs are interchangeable.

\begin{theorem}\label{th:probabilistic algo}
    Let $(G,I,O)$ be a labelled open graph and define
    \[
        \Lambda = \{\lambda:\comp{O}\to\{\XY,\XZ,\YZ\} \mid (G,I,O,\ld) \text{ has gflow}\}.
    \]
    There exists a probabilistic algorithm that outputs any $\lambda\in\Lambda$ with positive probability.
\end{theorem}

\begin{proof}
    We modify the algorithm from Section~\ref{sec:algorithm}. Whenever solving a layer and identifying $\New$ in line~\ref{Line:New from Layers}, instead of solving all such vertices, pick a single vertex $v \in \New$. If $v$ has more than one measurement plane associated with it in $\Found$, then choose a valid label $\ld(v)$ at random; otherwise pick the unique $\ld(v)$ recorded in $\Found$. After solving for $v$, instead of continuing for the remaining vertices in $\New$, immediately move to the next layer.

    Suppose $\ld\in\Lambda$.
    Let $(c,\prec)$ be a gflow on $(G,I,O,\ld)$, and let $<$ be any totalisation of $\prec$. There is a non-zero probability that the algorithm picks non-output vertices precisely in descending order according to $<$. If any vertex is triple-correctable, the algorithm also has a non-zero probability of choosing the same labelling as in $\ld$; otherwise, it must pick the measurement label that appears in $\ld$ with certainty. Therefore, the presented probabilistic approach has a non-zero probability of choosing the branch that yields $\ld$, as desired.
\end{proof}

The algorithm is less efficient than the one presented in Section~\ref{sec:algorithm}, but it still achieves $\bigO(n^3)$ complexity while allowing access to any $\ld$.

By exploring the full decision tree encountered by the algorithm in Theorem~\ref{th:probabilistic algo}, one can enumerate all possible measurement labellings that result in gflow for a given open graph.
This full exploration may require exponential time in the size of the initial graph, \ie it is no longer efficient. Furthermore, different branches may reach the same labelling.
Nevertheless, this approach may be useful when considering the question of which one-way computations are supported on some given graph state.

\section{Conclusions}
\label{sec:conclusions}

Where previous work on flow properties and flow-finding focused on labelled open graphs where the measurement labelling is fixed, we instead consider the scenario where the measurement labels may be unknown or only partially specified.
We give an algorithm that resolves the problem of finding gflow in this scenario in $\bigO(n^3)$, where $n$ is the number of qubits in the computation.
This is the same complexity as the best-known standard gflow-finding algorithm \cite{mitosekAlgebraicInterpretationPauli2026}, which our work strictly generalises.
Without affecting the complexity, our algorithm can also take in constraints on the order in which the measurements are to be performed.

We moreover show that if an open graph is balanced -- i.e.\ the number of inputs equals the number of outputs -- then there exists at most one measurement labelling compatible with gflow.
This generalises the property that a special variant of gflow called focused is unique on balanced labelled open graphs (if it exists) \cite{mhallaWhichGraphStates2014a}, which is often useful when working with flow properties.
As a next step we show that if an unbalanced open graph (labelled or not) has gflow, one can efficiently identify a set of non-input vertices that can be turned into inputs to balance the open graph while preserving the existence of flow and requiring only minimal changes to the measurement labelling.
In particular, this means that given only a graph and a set of output vertices (plus potentially some constraints on the time order or some partial labelling constraints), we can efficiently find a measurement labelling and a set of inputs such that the resulting labelled open graph is balanced and has gflow (or conclude no such labelling exists).
This is a significant step towards addressing the question of which robustly deterministic one-way computations can be implemented given some graph state that can be experimentally prepared.

Flow structures were also defined outside of qubit computation with $\mathbb{F}_d$-flow (also called $\mathbb{Z}_d$-flow) on qudits of odd prime dimension $d$ \cite{boothMeasurementbasedQuantumComputation2022, boothOutcomeDeterminismMeasurementbased2023,mitosekWorking2026}.
There also exists a continuous variable flow \cite{boothMeasurementbasedQuantumComputation2022, BoothCVflow2023}.
An adaptation of gflow to hypergraph states was demonstrated in \cite{IJcken2024Masters}.
In \cite{deFeliceDataflow2026}, a partial static flow was considered as a special type of Pauli flow for more efficient reasoning for fusion networks.
Another adaptation of Pauli flow is ZX-flow, which can be defined on a broader class of ZX-calculus diagrams \cite{kissinger2026zxflowflexiblecriteriondeterministic}.
All of these offer directions in which our results could be generalised in the future.

\section*{Acknowledgements}
We thank Simon Perdrix for interesting and useful discussions.

VG and PM contributed equally.
VG found the initial version of the algorithm, which RR and PM also found independently.
PM proved the uniqueness of the measurement labelling for balanced open graphs.
MB, PM, and VG together optimised the main algorithm (building on the initial version by VG) and MB devised a balancing procedure.
RR provided physical insights.
PM and MB wrote the paper, with notes from VG and RR.

MB is supported by the Plan France 2030 through the PEPR integrated project EPiQ ANR-22-PETQ-0007 and the HQI platform ANR-22-PNCQ-0002.
PM is supported by the Niedersächsisches Ministerium für Wissenschaft und Kultur.
RR and PM are supported by the Alexander von Humboldt Foundation.

\begin{sloppypar}
\bibliographystyle{eptcs}
\bibliography{libref}
\end{sloppypar}

\newpage
\appendix

\section{Worked example of labelling-finding}\label{sec:worked example}

We present an example of the assignments our algorithm makes to two closely related open graphs.

\subsection{Balanced open graph}\label{subsec:worked example balanced}
We consider the open graph from Figure~\ref{fig:uniqueness-example}, now without the labelling from that figure:
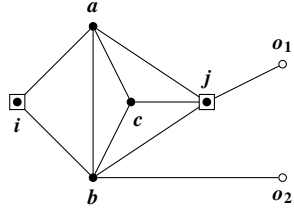
\begin{align*}
    \begin{tikzpicture}
    \begin{pgfonlayer}{nodelayer}
        \node [style=ZH H] (18) at (-3, 0) {};
        \node [style=ZH H] (19) at (2, 0) {};
        \node [style=GR Tiny Black] (0) at (-3, 0) {};   
        \node [style=GR Tiny Black] (1) at (-1, 2) {};   
        \node [style=GR Tiny Black] (2) at (-1, -2) {};  
        \node [style=GR Tiny Black] (3) at (0, 0) {};    
        \node [style=GR Tiny Black] (4) at (2, 0) {};    
        \node [style=GR Tiny Empty] (5) at (4, 1) {};    
        \node [style=GR Tiny Empty] (6) at (4, -2) {};   
        \node [style=BlackTEXT] (7) at (-3, -0.6) {$i$};
        \node [style=BlackTEXT] (8) at (-1, 2.5) {$a$};
        \node [style=BlackTEXT] (9) at (-1, -2.5) {$b$};
        \node [style=BlackTEXT] (10) at (0.2, -0.5) {$c$};
        \node [style=BlackTEXT] (11) at (2, 0.6) {$j$};
        \node [style=BlackTEXT] (12) at (4, 1.5) {$o_1$};
        \node [style=BlackTEXT] (13) at (4, -2.5) {$o_2$};
        \node [style=none] (14) at (-4.5, 0) {};
        \node [style=none] (15) at (4.5, 0) {};
        \node [style=none] (16) at (0, 3.5) {};
        \node [style=none] (17) at (0, -3.5) {};
    \end{pgfonlayer}
    \begin{pgfonlayer}{edgelayer}
        \draw (0) to (1); 
        \draw (0) to (2); 
        \draw (1) to (2); 
        \draw (1) to (3); 
        \draw (1) to (4); 
        \draw (2) to (3); 
        \draw (2) to (4); 
        \draw (3) to (4); 
        \draw (4) to (5); 
        \draw (2) to (6); 
    \end{pgfonlayer}
\end{tikzpicture}
\end{align*}
As this is a balanced open graph, by Theorem~\ref{Th:uniqueness} we will expect the algorithm to recover the same labelling.

The initial system of linear equations constructed for this open graph is:
\begin{equation*}
    \begin{pNiceArray}{ccc:cc||c|c|ccc|ccc|ccc||ccccc}[first-row, first-col]
        & a & b & c & o_1 & o_2
        & i_{\XY} & j_{\XY}
        & a_{\XY} & a_{\YZ} & a_{\XZ}
        & b_{\XY} & b_{\YZ} & b_{\XZ}
        & c_{\XY} & c_{\YZ} & c_{\XZ}
        & i' & j' & a' & b' & c' \\
    i   & 1 & 1 & 0 & 0 & 0
        & 1 & 0
        & 0 & 1 & 1
        & 0 & 1 & 1
        & 0 & 0 & 0
        & 1 & 0 & 0 & 0 & 0 \\
    j   & 1 & 1 & 1 & 1 & 0
        & 0 & 1
        & 0 & 1 & 1
        & 0 & 1 & 1
        & 0 & 1 & 1
        & 0 & 1 & 0 & 0 & 0 \\
    \hdottedline
    a   & 0 & 1 & 1 & 0 & 0
        & 0 & 0
        & 1 & 0 & 1
        & 0 & 1 & 1
        & 0 & 1 & 1
        & 0 & 0 & 1 & 0 & 0 \\
    b   & 1 & 0 & 1 & 0 & 1
        & 0 & 0
        & 0 & 1 & 1
        & 1 & 0 & 1
        & 0 & 1 & 1
        & 0 & 0 & 0 & 1 & 0 \\
    c   & 1 & 1 & 0 & 0 & 0
        & 0 & 0
        & 0 & 1 & 1
        & 0 & 1 & 1
        & 1 & 0 & 1
        & 0 & 0 & 0 & 0 & 1 \\
    \end{pNiceArray}
\end{equation*}
Recall that in this system, the first block (columns labelled $a,b,c,o_1,o_2$) corresponds to the coefficients of the linear equations.
There are five variables, one for each of the non-input vertices.
The second block (columns labelled $i_\XY$ to $c_\XZ$) represents the constants: in fact, we are here looking at 11 systems of linear equations that share the same coefficients but may differ in their constants.
The role of the third block (columns with primed labels) is to help in the process of updating the system after each `layer' of labellings has been assigned; it is otherwise irrelevant to the linear systems.
Since this matrix is quite large, we will make some simplifications to the algorithm given in Section~\ref{sec:algorithm} to make it easier to follow `by hand', in particular:
\begin{itemize}
    \item We do not re-order the columns once new variables are `activated', \cf Definition~\ref{def:active}.
    In the algorithm, columns are re-ordered so that the columns of coefficients corresponding to active variables -- those that are allowed to take value 1 given the labels that have already been found -- are at the beginning for faster access.
    Here, we just manually keep track of which variables are active.
    \item Once a solution has been found for a vertex, the algorithm sets a corresponding row in the matrix to zero by using the second copy of the data and the information encoded in the third block.
    We instead simply delete these rows.
    (It so happens that we will never be in the situation where multiple rows depend on the original row of a vertex we've just solved, so we do not need to deal with that scenario.)
\end{itemize}

For the first layer, only variables corresponding to outputs can take non-zero values; in the algorithm this means that Gaussian elimination is performed only with respect to these two variables.
In this case, it suffices to reorder the rows; we will nevertheless now drop the row labels as they lose their meaning once we create linear combinations of rows anyway.
Similarly, we drop the dotted lines.
\begin{equation*}
    \begin{pNiceArray}{ccccc||c|c|ccc|ccc|ccc||ccccc}[first-row, first-col]
        & a & b & c & o_1 & o_2
        & i_{\XY} & j_{\XY}
        & a_{\XY} & a_{\YZ} & a_{\XZ}
        & b_{\XY} & b_{\YZ} & b_{\XZ}
        & c_{\XY} & c_{\YZ} & c_{\XZ}
        & i' & j' & a' & b' & c' \\
       & 1 & 1 & 1 & 1 & 0
        & 0 & 1
        & 0 & 1 & 1
        & 0 & 1 & 1
        & 0 & 1 & 1
        & 0 & 1 & 0 & 0 & 0 \\
       & 1 & 0 & 1 & 0 & 1
        & 0 & 0
        & 0 & 1 & 1
        & 1 & 0 & 1
        & 0 & 1 & 1
        & 0 & 0 & 0 & 1 & 0 \\
       & 1 & 1 & 0 & 0 & 0
        & 1 & 0
        & 0 & 1 & 1
        & 0 & 1 & 1
        & 0 & 0 & 0
        & 1 & 0 & 0 & 0 & 0 \\
       & 0 & 1 & 1 & 0 & 0
        & 0 & 0
        & 1 & 0 & 1
        & 0 & 1 & 1
        & 0 & 1 & 1
        & 0 & 0 & 1 & 0 & 0 \\
       & 1 & 1 & 0 & 0 & 0
        & 0 & 0
        & 0 & 1 & 1
        & 0 & 1 & 1
        & 1 & 0 & 1
        & 0 & 0 & 0 & 0 & 1 \\
    \end{pNiceArray}
\end{equation*}
The first round of the algorithm presented in Section~\ref{sec:algorithm} solves the initial set of vertices as follows.
Note that only $o_1,o_2$ may be non-zero, so the only vertices that can be solved are ones for which there exists a column of constants that ends in three zeroes.
In this case this is true for the following columns:
\begin{itemize}
    \item Column $j_\XY$ satisfies the property with a single 1 in the same row as column $o_1$.
    We can thus assign label $\XY$ to vertex $j$, witnessed by the correction set $g(j) = \{ o_1 \}$.

    \item Column $b_\XY$ satisfies the property with a single 1 in the same row as column $o_2$.
    We can thus assign label$\XY$ to vertex $b$, witnessed by the correction set $g(b) = \{ o_2 \}$.
\end{itemize}
This means that after the first layer, the partially labelled open graph is:
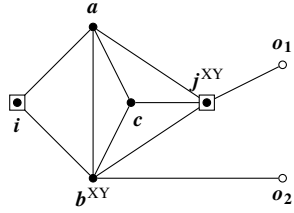
\begin{align*}
    \begin{tikzpicture}
    \begin{pgfonlayer}{nodelayer}
        \node [style=ZH H] (18) at (-3, 0) {};
        \node [style=ZH H] (19) at (2, 0) {};
        \node [style=GR Tiny Black] (0) at (-3, 0) {};   
        \node [style=GR Tiny Black] (1) at (-1, 2) {};   
        \node [style=GR Tiny Black] (2) at (-1, -2) {};  
        \node [style=GR Tiny Black] (3) at (0, 0) {};    
        \node [style=GR Tiny Black] (4) at (2, 0) {};    
        \node [style=GR Tiny Empty] (5) at (4, 1) {};    
        \node [style=GR Tiny Empty] (6) at (4, -2) {};   
        \node [style=BlackTEXT] (7) at (-3, -0.6) {$i$};
        \node [style=BlackTEXT] (8) at (-1, 2.5) {$a$};
        \node [style=BlackTEXT] (9) at (-1, -2.5) {$b^{\XY}$};
        \node [style=BlackTEXT] (10) at (0.2, -0.5) {$c$};
        \node [style=BlackTEXT] (11) at (2, 0.6) {$j^{\XY}$};
        \node [style=BlackTEXT] (12) at (4, 1.5) {$o_1$};
        \node [style=BlackTEXT] (13) at (4, -2.5) {$o_2$};
        \node [style=none] (14) at (-4.5, 0) {};
        \node [style=none] (15) at (4.5, 0) {};
        \node [style=none] (16) at (0, 3.5) {};
        \node [style=none] (17) at (0, -3.5) {};
    \end{pgfonlayer}
    \begin{pgfonlayer}{edgelayer}
        \draw (0) to (1); 
        \draw (0) to (2); 
        \draw (1) to (2); 
        \draw (1) to (3); 
        \draw (1) to (4); 
        \draw (2) to (3); 
        \draw (2) to (4); 
        \draw (3) to (4); 
        \draw (4) to (5); 
        \draw (2) to (6); 
    \end{pgfonlayer}
\end{tikzpicture}
\end{align*}
From now on, we no longer care about side effects on $b$ or $j$, so the corresponding rows -- which are marked by having a 1 in column $b'$ or $j'$ -- are removed.
For simplicity, we will also drop all columns in the second and third block with labels derived from $b$ or $j$, as these are no longer relevant.
\begin{equation*}
    \begin{pNiceArray}{ccccc||c|ccc|ccc||ccc}[first-row, first-col]
        & a & b & c & o_1 & o_2
        & i_{\XY}
        & a_{\XY} & a_{\YZ} & a_{\XZ}
        & c_{\XY} & c_{\YZ} & c_{\XZ}
        & i' & a' & c' \\
       & 1 & 1 & 0 & 0 & 0
        & 1
        & 0 & 1 & 1
        & 0 & 0 & 0
        & 1 & 0 & 0 \\
       & 0 & 1 & 1 & 0 & 0
        & 0
        & 1 & 0 & 1
        & 0 & 1 & 1
        & 0 & 1 & 0 \\
       & 1 & 1 & 0 & 0 & 0
        & 0
        & 0 & 1 & 1
        & 1 & 0 & 1
        & 0 & 0 & 1 \\
    \end{pNiceArray}
\end{equation*}
Now, the variable corresponding to $b$ becomes available so a further partial Gaussian elimination is performed; in this case we add the top row to each of the other two rows.
\begin{equation*}
    \begin{pNiceArray}{ccccc||c|ccc|ccc||ccc}[first-row, first-col]
        & a & b & c & o_1 & o_2
        & i_{\XY}
        & a_{\XY} & a_{\YZ} & a_{\XZ}
        & c_{\XY} & c_{\YZ} & c_{\XZ}
        & i' & a' & c' \\
       & 1 & 1 & 0 & 0 & 0
        & 1
        & 0 & 1 & 1
        & 0 & 0 & 0
        & 1 & 0 & 0 \\
       & 1 & 0 & 1 & 0 & 0
        & 1
        & 1 & 1 & 0
        & 0 & 1 & 1
        & 1 & 1 & 0 \\
       & 0 & 0 & 0 & 0 & 0
        & 1
        & 0 & 0 & 0
        & 1 & 0 & 1
        & 1 & 0 & 1 \\
    \end{pNiceArray}
\end{equation*}
Since columns $o_1$ and $o_2$ are identically 0 (so that vertices $o_1$ and $o_2$ can no longer contribute to corrections), we now need to find a column in the second block that is zero everywhere except possibly the top row.
This holds for column $a_\XZ$, so vertex $a$ is detected as solvable in the plane $\XZ$.
The corresponding correction set is $g(a) = \{ a,b \}$ (where the vertex itself is added to the solution of the linear system because it gets a label in $\{\XZ,\YZ\}$), leading to:
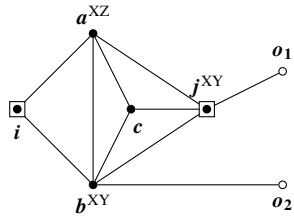
\begin{align*}
    \begin{tikzpicture}
    \begin{pgfonlayer}{nodelayer}
        \node [style=ZH H] (18) at (-3, 0) {};
        \node [style=ZH H] (19) at (2, 0) {};
        \node [style=GR Tiny Black] (0) at (-3, 0) {};   
        \node [style=GR Tiny Black] (1) at (-1, 2) {};   
        \node [style=GR Tiny Black] (2) at (-1, -2) {};  
        \node [style=GR Tiny Black] (3) at (0, 0) {};    
        \node [style=GR Tiny Black] (4) at (2, 0) {};    
        \node [style=GR Tiny Empty] (5) at (4, 1) {};    
        \node [style=GR Tiny Empty] (6) at (4, -2) {};   
        \node [style=BlackTEXT] (7) at (-3, -0.6) {$i$};
        \node [style=BlackTEXT] (8) at (-1, 2.5) {$a^{\XZ}$};
        \node [style=BlackTEXT] (9) at (-1, -2.5) {$b^{\XY}$};
        \node [style=BlackTEXT] (10) at (0.2, -0.5) {$c$};
        \node [style=BlackTEXT] (11) at (2, 0.6) {$j^{\XY}$};
        \node [style=BlackTEXT] (12) at (4, 1.5) {$o_1$};
        \node [style=BlackTEXT] (13) at (4, -2.5) {$o_2$};
        \node [style=none] (14) at (-4.5, 0) {};
        \node [style=none] (15) at (4.5, 0) {};
        \node [style=none] (16) at (0, 3.5) {};
        \node [style=none] (17) at (0, -3.5) {};
    \end{pgfonlayer}
    \begin{pgfonlayer}{edgelayer}
        \draw (0) to (1); 
        \draw (0) to (2); 
        \draw (1) to (2); 
        \draw (1) to (3); 
        \draw (1) to (4); 
        \draw (2) to (3); 
        \draw (2) to (4); 
        \draw (3) to (4); 
        \draw (4) to (5); 
        \draw (2) to (6); 
    \end{pgfonlayer}
\end{tikzpicture}
\end{align*}
There is a single 1 in column $a'$, indicating that only the middle row depends on the original row $a$.
We can thus drop this row, and again eliminate the columns in blocks two and three whose labels are derived from $a$.
\begin{equation*}
    \begin{pNiceArray}{ccccc||c|ccc|ccc||ccc}[first-row, first-col]
        & a & b & c & o_1 & o_2
        & i_{\XY}
        & c_{\XY} & c_{\YZ} & c_{\XZ}
        & i' & c' \\
       & 1 & 1 & 0 & 0 & 0
        & 1
        & 0 & 0 & 0
        & 1 & 0 \\
       & 0 & 0 & 0 & 0 & 0
        & 1
        & 1 & 0 & 1
        & 1 & 1 \\
    \end{pNiceArray}
\end{equation*}
Variable $a$ now becomes available.
The following partial Gaussian elimination step is nevertheless trivial.
Next, vertex $c$ can be solved in $\YZ$ with $g(c) = \{ c \}$:
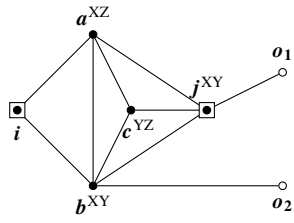
\begin{align*}
    \begin{tikzpicture}
    \begin{pgfonlayer}{nodelayer}
        \node [style=ZH H] (18) at (-3, 0) {};
        \node [style=ZH H] (19) at (2, 0) {};
        \node [style=GR Tiny Black] (0) at (-3, 0) {};   
        \node [style=GR Tiny Black] (1) at (-1, 2) {};   
        \node [style=GR Tiny Black] (2) at (-1, -2) {};  
        \node [style=GR Tiny Black] (3) at (0, 0) {};    
        \node [style=GR Tiny Black] (4) at (2, 0) {};    
        \node [style=GR Tiny Empty] (5) at (4, 1) {};    
        \node [style=GR Tiny Empty] (6) at (4, -2) {};   
        \node [style=BlackTEXT] (7) at (-3, -0.6) {$i$};
        \node [style=BlackTEXT] (8) at (-1, 2.5) {$a^{\XZ}$};
        \node [style=BlackTEXT] (9) at (-1, -2.5) {$b^{\XY}$};
        \node [style=BlackTEXT] (10) at (0.2, -0.5) {$c^{\YZ}$};
        \node [style=BlackTEXT] (11) at (2, 0.6) {$j^{\XY}$};
        \node [style=BlackTEXT] (12) at (4, 1.5) {$o_1$};
        \node [style=BlackTEXT] (13) at (4, -2.5) {$o_2$};
        \node [style=none] (14) at (-4.5, 0) {};
        \node [style=none] (15) at (4.5, 0) {};
        \node [style=none] (16) at (0, 3.5) {};
        \node [style=none] (17) at (0, -3.5) {};
    \end{pgfonlayer}
    \begin{pgfonlayer}{edgelayer}
        \draw (0) to (1); 
        \draw (0) to (2); 
        \draw (1) to (2); 
        \draw (1) to (3); 
        \draw (1) to (4); 
        \draw (2) to (3); 
        \draw (2) to (4); 
        \draw (3) to (4); 
        \draw (4) to (5); 
        \draw (2) to (6); 
    \end{pgfonlayer}
\end{tikzpicture}
\end{align*}
Again, there is a single 1 in column $c'$ so we can just drop the corresponding row and (for simplicity) the columns with labels derived from $c$.
\begin{equation*}
    \begin{pNiceArray}{ccccc||c|ccc|ccc||ccc}[first-row, first-col]
        & a & b & c & o_1 & o_2
        & i_{\XY}
        & i' \\
       & 1 & 1 & 0 & 0 & 0
        & 1
        & 1 \\
    \end{pNiceArray}
\end{equation*}
Lastly, vertex $i$ can be solved in $\XY$, $g(i)=\{ b \}$ recovering the labelled open graph from Figure~\ref{fig:uniqueness-example}:
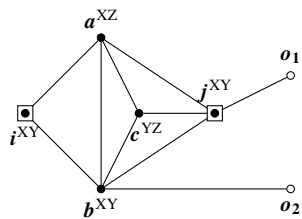
\begin{align*}
    \begin{tikzpicture}
	\begin{pgfonlayer}{nodelayer}
		\node [style=ZH H] (18) at (-3, 0) {};
        \node [style=ZH H] (19) at (2, 0) {};
		\node [style=GR Tiny Black] (0) at (-3, 0) {};   
		\node [style=GR Tiny Black] (1) at (-1, 2) {};   
		\node [style=GR Tiny Black] (2) at (-1, -2) {};  
		\node [style=GR Tiny Black] (3) at (0, 0) {};    
		\node [style=GR Tiny Black] (4) at (2, 0) {};    
		\node [style=GR Tiny Empty] (5) at (4, 1) {};    
		\node [style=GR Tiny Empty] (6) at (4, -2) {};   
		\node [style=BlackTEXT] (7) at (-3, -0.6) {$i^{\XY}$};
		\node [style=BlackTEXT] (8) at (-1, 2.5) {$a^{\XZ}$};
		\node [style=BlackTEXT] (9) at (-1, -2.5) {$b^{\XY}$};
		\node [style=BlackTEXT] (10) at (0.2, -0.5) {$c^{\YZ}$};
		\node [style=BlackTEXT] (11) at (2, 0.6) {$j^{\XY}$};
		\node [style=BlackTEXT] (12) at (4, 1.5) {$o_1$};
		\node [style=BlackTEXT] (13) at (4, -2.5) {$o_2$};
		\node [style=none] (14) at (-4.5, 0) {};
		\node [style=none] (15) at (4.5, 0) {};
		\node [style=none] (16) at (0, 3.5) {};
		\node [style=none] (17) at (0, -3.5) {};
	\end{pgfonlayer}
	\begin{pgfonlayer}{edgelayer}
		\draw (0) to (1); 
		\draw (0) to (2); 
		\draw (1) to (2); 
		\draw (1) to (3); 
		\draw (1) to (4); 
		\draw (2) to (3); 
		\draw (2) to (4); 
		\draw (3) to (4); 
		\draw (4) to (5); 
		\draw (2) to (6); 
	\end{pgfonlayer}
\end{tikzpicture}
\end{align*}
Alternatively, $g(i) = \{a\}$ would also have been an option here, but this is just a change of the correction set: the label remains unaffected.

The layers in which vertices were assigned a plane are: $b, j \mid a \mid c \mid i$.

\subsection{Unbalanced labelled open graph}\label{subsec:worked example unbalanced}

If the starting labelled open graph did not have $j$ specified as an input, that is:
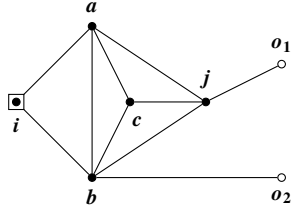
\begin{align*}
    \begin{tikzpicture}
    \begin{pgfonlayer}{nodelayer}
        \node [style=ZH H] (18) at (-3, 0) {};
        \node [style=GR Tiny Black] (0) at (-3, 0) {};   
        \node [style=GR Tiny Black] (1) at (-1, 2) {};   
        \node [style=GR Tiny Black] (2) at (-1, -2) {};  
        \node [style=GR Tiny Black] (3) at (0, 0) {};    
        \node [style=GR Tiny Black] (4) at (2, 0) {};    
        \node [style=GR Tiny Empty] (5) at (4, 1) {};    
        \node [style=GR Tiny Empty] (6) at (4, -2) {};   
        \node [style=BlackTEXT] (7) at (-3, -0.6) {$i$};
        \node [style=BlackTEXT] (8) at (-1, 2.5) {$a$};
        \node [style=BlackTEXT] (9) at (-1, -2.5) {$b$};
        \node [style=BlackTEXT] (10) at (0.2, -0.5) {$c$};
        \node [style=BlackTEXT] (11) at (2, 0.6) {$j$};
        \node [style=BlackTEXT] (12) at (4, 1.5) {$o_1$};
        \node [style=BlackTEXT] (13) at (4, -2.5) {$o_2$};
        \node [style=none] (14) at (-4.5, 0) {};
        \node [style=none] (15) at (4.5, 0) {};
        \node [style=none] (16) at (0, 3.5) {};
        \node [style=none] (17) at (0, -3.5) {};
    \end{pgfonlayer}
    \begin{pgfonlayer}{edgelayer}
        \draw (0) to (1); 
        \draw (0) to (2); 
        \draw (1) to (2); 
        \draw (1) to (3); 
        \draw (1) to (4); 
        \draw (2) to (3); 
        \draw (2) to (4); 
        \draw (3) to (4); 
        \draw (4) to (5); 
        \draw (2) to (6); 
    \end{pgfonlayer}
\end{tikzpicture}
\end{align*}
then the algorithm would find a different solution with fewer layers.

The initial system of linear equations constructed for this modified open graph is the following, which differs from the one in the previous subsection by having a column labelled $j$ in the first block and two extra columns $j_\YZ,j_\XZ$ in the second block.
\begin{equation*}
	\scalemath{0.9}{
	\begin{pNiceArray}{cccc:cc||c|ccc|ccc|ccc|ccc||ccccc}[first-row, first-col]
		& j & a & b & c & o_1 & o_2
		& i_{\XY}
		& j_{\XY} & j_{\YZ} & j_{\XZ}
		& a_{\XY} & a_{\YZ} & a_{\XZ}
		& b_{\XY} & b_{\YZ} & b_{\XZ}
		& c_{\XY} & c_{\YZ} & c_{\XZ}
		& i' & j' & a' & b' & c' \\
	i   & 0 & 1 & 1 & 0 & 0 & 0
		& 1 
		& 0 & 0 & 0
		& 0 & 1 & 1
		& 0 & 1 & 1
		& 0 & 0 & 0
		& 1 & 0 & 0 & 0 & 0 \\
		\hdottedline
	j   & 0 & 1 & 1 & 1 & 1 & 0
		& 0
		& 1 & 0 & 1
		& 0 & 1 & 1
		& 0 & 1 & 1
		& 0 & 1 & 1
		& 0 & 1 & 0 & 0 & 0 \\
	a   & 1 & 0 & 1 & 1 & 0 & 0
		& 0 
		& 0 & 1 & 1
		& 1 & 0 & 1
		& 0 & 1 & 1
		& 0 & 1 & 1
		& 0 & 0 & 1 & 0 & 0 \\
	b   & 1 & 1 & 0 & 1 & 0 & 1
		& 0 
		& 0 & 1 & 1
		& 0 & 1 & 1
		& 1 & 0 & 1
		& 0 & 1 & 1
		& 0 & 0 & 0 & 1 & 0 \\
	c   & 1 & 1 & 1 & 0 & 0 & 0
		& 0 
		& 0 & 1 & 1
		& 0 & 1 & 1
		& 0 & 1 & 1
		& 1 & 0 & 1
		& 0 & 0 & 0 & 0 & 1 \\
	\end{pNiceArray}
	}
\end{equation*}
As before, we begin by bringing the columns $o_1$ and $o_2$ into echelon form, which can be done simply by permuting rows.
\begin{equation*}
	\scalemath{0.9}{
		\begin{pNiceArray}{cccccc||c|ccc|ccc|ccc|ccc||ccccc}[first-row, first-col]
			& j & a & b & c & o_1 & o_2
			& i_{\XY}
			& j_{\XY} & j_{\YZ} & j_{\XZ}
			& a_{\XY} & a_{\YZ} & a_{\XZ}
			& b_{\XY} & b_{\YZ} & b_{\XZ}
			& c_{\XY} & c_{\YZ} & c_{\XZ}
			& i' & j' & a' & b' & c' \\
			& 0 & 1 & 1 & 1 & 1 & 0
			& 0
			& 1 & 0 & 1
			& 0 & 1 & 1
			& 0 & 1 & 1
			& 0 & 1 & 1
			& 0 & 1 & 0 & 0 & 0 \\
			& 1 & 1 & 0 & 1 & 0 & 1
			& 0 
			& 0 & 1 & 1
			& 0 & 1 & 1
			& 1 & 0 & 1
			& 0 & 1 & 1
			& 0 & 0 & 0 & 1 & 0 \\
			& 0 & 1 & 1 & 0 & 0 & 0
			& 1 
			& 0 & 0 & 0
			& 0 & 1 & 1
			& 0 & 1 & 1
			& 0 & 0 & 0
			& 1 & 0 & 0 & 0 & 0 \\
			& 1 & 0 & 1 & 1 & 0 & 0
			& 0 
			& 0 & 1 & 1
			& 1 & 0 & 1
			& 0 & 1 & 1
			& 0 & 1 & 1
			& 0 & 0 & 1 & 0 & 0 \\
			& 1 & 1 & 1 & 0 & 0 & 0
			& 0 
			& 0 & 1 & 1
			& 0 & 1 & 1
			& 0 & 1 & 1
			& 1 & 0 & 1
			& 0 & 0 & 0 & 0 & 1 \\
		\end{pNiceArray}
	}
\end{equation*}
Again $b$ and $j$ can be assigned $\XY$ and solved with $g(b) = \{ o_2 \}$ and $g(j) = \{ o_1 \}$, as in Subsection~\ref{subsec:worked example balanced}:
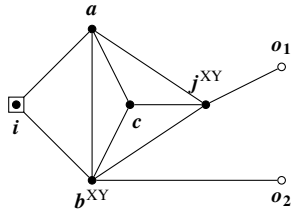
\begin{align*}
    \begin{tikzpicture}
    \begin{pgfonlayer}{nodelayer}
        \node [style=ZH H] (18) at (-3, 0) {};
        \node [style=GR Tiny Black] (0) at (-3, 0) {};   
        \node [style=GR Tiny Black] (1) at (-1, 2) {};   
        \node [style=GR Tiny Black] (2) at (-1, -2) {};  
        \node [style=GR Tiny Black] (3) at (0, 0) {};    
        \node [style=GR Tiny Black] (4) at (2, 0) {};    
        \node [style=GR Tiny Empty] (5) at (4, 1) {};    
        \node [style=GR Tiny Empty] (6) at (4, -2) {};   
        \node [style=BlackTEXT] (7) at (-3, -0.6) {$i$};
        \node [style=BlackTEXT] (8) at (-1, 2.5) {$a$};
        \node [style=BlackTEXT] (9) at (-1, -2.5) {$b^{\XY}$};
        \node [style=BlackTEXT] (10) at (0.2, -0.5) {$c$};
        \node [style=BlackTEXT] (11) at (2, 0.6) {$j^{\XY}$};
        \node [style=BlackTEXT] (12) at (4, 1.5) {$o_1$};
        \node [style=BlackTEXT] (13) at (4, -2.5) {$o_2$};
        \node [style=none] (14) at (-4.5, 0) {};
        \node [style=none] (15) at (4.5, 0) {};
        \node [style=none] (16) at (0, 3.5) {};
        \node [style=none] (17) at (0, -3.5) {};
    \end{pgfonlayer}
    \begin{pgfonlayer}{edgelayer}
        \draw (0) to (1); 
        \draw (0) to (2); 
        \draw (1) to (2); 
        \draw (1) to (3); 
        \draw (1) to (4); 
        \draw (2) to (3); 
        \draw (2) to (4); 
        \draw (3) to (4); 
        \draw (4) to (5); 
        \draw (2) to (6); 
    \end{pgfonlayer}
\end{tikzpicture}
\end{align*}
We can now drop the rows corresponding to $b$ and $j$ (which rows these are can be read off from the third block), together with any second- and third-block columns with labels derived from either vertex.
\begin{equation*}
		\begin{pNiceArray}{cccccc||c|ccc|ccc||ccccc}[first-row, first-col]
			& j & a & b & c & o_1 & o_2
			& i_{\XY}
			& a_{\XY} & a_{\YZ} & a_{\XZ}
			& c_{\XY} & c_{\YZ} & c_{\XZ}
			& i' & a' & c' \\
			& 0 & 1 & 1 & 0 & 0 & 0
			& 1 
			& 0 & 1 & 1
			& 0 & 0 & 0
			& 1 & 0 & 0 \\
			& 1 & 0 & 1 & 1 & 0 & 0
			& 0 
			& 1 & 0 & 1
			& 0 & 1 & 1
			& 0 & 1 & 0 \\
			& 1 & 1 & 1 & 0 & 0 & 0
			& 0 
			& 0 & 1 & 1
			& 1 & 0 & 1
			& 0 & 0 & 1 \\
		\end{pNiceArray}
\end{equation*}
As $j$ is no longer an input, the algorithm can now use $j$ as well as $b$ in correction sets.
We bring the matrix into echelon form with respect to columns $j,b,o_1,o_2$ by
\begin{itemize}
	\item Switching the top two rows,
	\item adding the new top row to the bottom one, and
	\item adding the new middle row to the top one.
\end{itemize}
\begin{equation*}
	\begin{pNiceArray}{cccccc||c|ccc|ccc||ccccc}[first-row, first-col]
		& j & a & b & c & o_1 & o_2
		& i_{\XY}
		& a_{\XY} & a_{\YZ} & a_{\XZ}
		& c_{\XY} & c_{\YZ} & c_{\XZ}
		& i' & a' & c' \\
		& 1 & 1 & 0 & 1 & 0 & 0
		& 1 
		& 1 & 1 & 0
		& 0 & 1 & 1
		& 1 & 1 & 0 \\
		& 0 & 1 & 1 & 0 & 0 & 0
		& 1 
		& 0 & 1 & 1
		& 0 & 0 & 0
		& 1 & 0 & 0 \\
		& 0 & 1 & 0 & 1 & 0 & 0
		& 0 
		& 1 & 1 & 0
		& 1 & 1 & 0
		& 0 & 1 & 1 \\
	\end{pNiceArray}
\end{equation*}
This means any columns in the second block that have a 0 on the third row become solvable.
Thus all remaining vertices can be solved in the second layer with:\begin{align*}
    \ld(i) &= \XY, & g(i) &= \{b, j\};\\
    \ld(a) &= \XZ, & g(a) &= \{a, b\}; \text{and}\\
    \ld(c) &= \XZ, & g(c) &= \{c, j\}.
\end{align*}
The resulting labelled open graph is:
\begin{align*}
    \begin{tikzpicture}
    \begin{pgfonlayer}{nodelayer}
        \node [style=ZH H] (18) at (-3, 0) {};
        \node [style=GR Tiny Black] (0) at (-3, 0) {};   
        \node [style=GR Tiny Black] (1) at (-1, 2) {};   
        \node [style=GR Tiny Black] (2) at (-1, -2) {};  
        \node [style=GR Tiny Black] (3) at (0, 0) {};    
        \node [style=GR Tiny Black] (4) at (2, 0) {};    
        \node [style=GR Tiny Empty] (5) at (4, 1) {};    
        \node [style=GR Tiny Empty] (6) at (4, -2) {};   
        \node [style=BlackTEXT] (7) at (-3, -0.6) {$i^{\XY}$};
        \node [style=BlackTEXT] (8) at (-1, 2.5) {$a^{\XZ}$};
        \node [style=BlackTEXT] (9) at (-1, -2.5) {$b^{\XY}$};
        \node [style=BlackTEXT] (10) at (0.2, -0.5) {$c^{\XZ}$};
        \node [style=BlackTEXT] (11) at (2, 0.6) {$j^{\XY}$};
        \node [style=BlackTEXT] (12) at (4, 1.5) {$o_1$};
        \node [style=BlackTEXT] (13) at (4, -2.5) {$o_2$};
        \node [style=none] (14) at (-4.5, 0) {};
        \node [style=none] (15) at (4.5, 0) {};
        \node [style=none] (16) at (0, 3.5) {};
        \node [style=none] (17) at (0, -3.5) {};
    \end{pgfonlayer}
    \begin{pgfonlayer}{edgelayer}
        \draw (0) to (1); 
        \draw (0) to (2); 
        \draw (1) to (2); 
        \draw (1) to (3); 
        \draw (1) to (4); 
        \draw (2) to (3); 
        \draw (2) to (4); 
        \draw (3) to (4); 
        \draw (4) to (5); 
        \draw (2) to (6); 
    \end{pgfonlayer}
\end{tikzpicture}
\end{align*}
The layers in which vertices were assigned a plane are: $b, j \mid i, a, c$.
Note that the contraction of layers required $\ld(c) = \XZ$ instead of $\YZ$, whereas the latter was necessary when $j$ was an input.

\subsection{Balancing procedure}\label{subsec:worked example balancing}
Suppose that we have the labelled open graph at the end of the previous subsection.
The gflow constructed by our algorithm is not necessarily focused; indeed, the stated gflow would need to be focused before proceeding with our balancing argument, resulting in the following correction function (note, that the correction sets of $b$ and $j$ did not require any changes):\begin{align*}
    g(i) &= \{b, j, o_1, o_2\},\\
    g(a) &= \{a, b, o_2\},\\
    g(c) &= \{c, j, o_1\},\\
    g(b) &= \{o_2\}, \text{and}\\
    g(j) &= \{o_1\}.
\end{align*}
This has a unique non-trivial focused set $\focusedset = \{j, o_2\}$.
Following the procedure after from Theorem~\ref{thm:balancing algo}, we find $\extfocset = \{ a, c, j, o_2 \}$ ($o_1$ is not included despite being in $\odd{\focusedset}$ as it is an output).
Since there are no other non-trivial focused sets, we have $\extallfoc = \{ a, c, j, o_2 \}$.
Among these, both $a$ and $c$ are minimal.
Turning either $a$ or $c$ into an input balances the open graph, resulting in the following labellings:\\
\noindent\makebox[\textwidth]{%
    \hfill
    \begin{tikzpicture}
    \begin{pgfonlayer}{nodelayer}
        \node [style=ZH H] (18) at (-3, 0) {};
        \node [style=ZH H] (19) at (-1, 2) {};
        \node [style=GR Tiny Black] (0) at (-3, 0) {};   
        \node [style=GR Tiny Black] (1) at (-1, 2) {};   
        \node [style=GR Tiny Black] (2) at (-1, -2) {};  
        \node [style=GR Tiny Black] (3) at (0, 0) {};    
        \node [style=GR Tiny Black] (4) at (2, 0) {};    
        \node [style=GR Tiny Empty] (5) at (4, 1) {};    
        \node [style=GR Tiny Empty] (6) at (4, -2) {};   
        \node [style=BlackTEXT] (7) at (-3, -0.6) {$i^{\XY}$};
        \node [style=BlackTEXT] (8) at (-1, 2.5) {$a^{\XY}$};
        \node [style=BlackTEXT] (9) at (-1, -2.5) {$b^{\XY}$};
        \node [style=BlackTEXT] (10) at (0.2, -0.5) {$c^{\XZ}$};
        \node [style=BlackTEXT] (11) at (2, 0.6) {$j^{\XY}$};
        \node [style=BlackTEXT] (12) at (4, 1.5) {$o_1$};
        \node [style=BlackTEXT] (13) at (4, -2.5) {$o_2$};
        \node [style=none] (14) at (-4.5, 0) {};
        \node [style=none] (15) at (4.5, 0) {};
        \node [style=none] (16) at (0, 3.5) {};
        \node [style=none] (17) at (0, -3.5) {};
    \end{pgfonlayer}
    \begin{pgfonlayer}{edgelayer}
        \draw (0) to (1); 
        \draw (0) to (2); 
        \draw (1) to (2); 
        \draw (1) to (3); 
        \draw (1) to (4); 
        \draw (2) to (3); 
        \draw (2) to (4); 
        \draw (3) to (4); 
        \draw (4) to (5); 
        \draw (2) to (6); 
    \end{pgfonlayer}
\end{tikzpicture}
    \hfill\hfill
    \begin{tikzpicture}
    \begin{pgfonlayer}{nodelayer}
        \node [style=ZH H] (18) at (-3, 0) {};
        \node [style=ZH H] (19) at (0, 0) {};
        \node [style=GR Tiny Black] (0) at (-3, 0) {};   
        \node [style=GR Tiny Black] (1) at (-1, 2) {};   
        \node [style=GR Tiny Black] (2) at (-1, -2) {};  
        \node [style=GR Tiny Black] (3) at (0, 0) {};    
        \node [style=GR Tiny Black] (4) at (2, 0) {};    
        \node [style=GR Tiny Empty] (5) at (4, 1) {};    
        \node [style=GR Tiny Empty] (6) at (4, -2) {};   
        \node [style=BlackTEXT] (7) at (-3, -0.6) {$i^{\XY}$};
        \node [style=BlackTEXT] (8) at (-1, 2.5) {$a^{\XZ}$};
        \node [style=BlackTEXT] (9) at (-1, -2.5) {$b^{\XY}$};
        \node [style=BlackTEXT] (10) at (0.2, -0.5) {$c^{\XY}$};
        \node [style=BlackTEXT] (11) at (2, 0.6) {$j^{\XY}$};
        \node [style=BlackTEXT] (12) at (4, 1.5) {$o_1$};
        \node [style=BlackTEXT] (13) at (4, -2.5) {$o_2$};
        \node [style=none] (14) at (-4.5, 0) {};
        \node [style=none] (15) at (4.5, 0) {};
        \node [style=none] (16) at (0, 3.5) {};
        \node [style=none] (17) at (0, -3.5) {};
    \end{pgfonlayer}
    \begin{pgfonlayer}{edgelayer}
        \draw (0) to (1); 
        \draw (0) to (2); 
        \draw (1) to (2); 
        \draw (1) to (3); 
        \draw (1) to (4); 
        \draw (2) to (3); 
        \draw (2) to (4); 
        \draw (3) to (4); 
        \draw (4) to (5); 
        \draw (2) to (6); 
    \end{pgfonlayer}
\end{tikzpicture}
    \hfill
}

Note that both are different from the original labelled open graph in Figure~\ref{fig:uniqueness-example} since obtaining that combination of labels and inputs would require making $j$ into an input and also changing the label of $c$ to $\YZ$.
Yet our balancing procedure is `local' in that it only changes the measurement label of a newly introduced input to $\XY$ (if needed), otherwise keeping the labelling unchanged.
In fact, if the labelling is allowed to be modified globally, it is also possible to find a \LOG that results in gflow while making $b$ into an input:
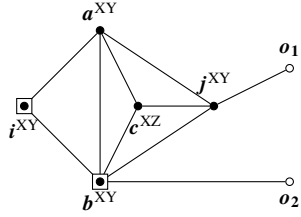
\begin{align*}
    \begin{tikzpicture}
    \begin{pgfonlayer}{nodelayer}
        \node [style=ZH H] (18) at (-3, 0) {};
        \node [style=ZH H] (19) at (-1, -2) {};
        \node [style=GR Tiny Black] (0) at (-3, 0) {};   
        \node [style=GR Tiny Black] (1) at (-1, 2) {};   
        \node [style=GR Tiny Black] (2) at (-1, -2) {};  
        \node [style=GR Tiny Black] (3) at (0, 0) {};    
        \node [style=GR Tiny Black] (4) at (2, 0) {};    
        \node [style=GR Tiny Empty] (5) at (4, 1) {};    
        \node [style=GR Tiny Empty] (6) at (4, -2) {};   
        \node [style=BlackTEXT] (7) at (-3, -0.6) {$i^{\XY}$};
        \node [style=BlackTEXT] (8) at (-1, 2.5) {$a^{\XY}$};
        \node [style=BlackTEXT] (9) at (-1, -2.5) {$b^{\XY}$};
        \node [style=BlackTEXT] (10) at (0.2, -0.5) {$c^{\XZ}$};
        \node [style=BlackTEXT] (11) at (2, 0.6) {$j^{\XY}$};
        \node [style=BlackTEXT] (12) at (4, 1.5) {$o_1$};
        \node [style=BlackTEXT] (13) at (4, -2.5) {$o_2$};
        \node [style=none] (14) at (-4.5, 0) {};
        \node [style=none] (15) at (4.5, 0) {};
        \node [style=none] (16) at (0, 3.5) {};
        \node [style=none] (17) at (0, -3.5) {};
    \end{pgfonlayer}
    \begin{pgfonlayer}{edgelayer}
        \draw (0) to (1); 
        \draw (0) to (2); 
        \draw (1) to (2); 
        \draw (1) to (3); 
        \draw (1) to (4); 
        \draw (2) to (3); 
        \draw (2) to (4); 
        \draw (3) to (4); 
        \draw (4) to (5); 
        \draw (2) to (6); 
    \end{pgfonlayer}
\end{tikzpicture}
\end{align*}
Yet neither $o_1$ nor $o_2$ can be turned into an input while preserving the property that the open graph has lgflow.

\section{Optimisation}\label{sec:implementation}

The presented algorithm achieves the target $\bigO(n^3)$ complexity.
However, an efficient implementation can include further optimisation tricks not included in the pseudocode.

We avoided including these optimisation tricks in the presented algorithm, as the vectors in the first and last blocks play a double role, which can be confusing:
\begin{enumerate}
    \item Firstly, observe that the second block of the constructed linear system contains separate vectors for each of the three possible measurement planes.
    It is possible to greatly reduce the attached block by observing that vectors corresponding to $\XY$ measurements are already present in the maintenance block of the linear system.
    Furthermore, the vectors corresponding to $\YZ$ measurements are also present but in the coefficient block instead.
    Thus, it is possible to maintain only the vectors corresponding to $\XZ$, while accessing information necessary for $\XY$ or $\YZ$ from the other two blocks.
    Furthermore, the vectors for $\XZ$ are simply the XOR of those for $\XY$ and $\YZ$ and could potentially also be dropped and reconstructed only when necessary.
    This way, we can drop two-thirds of the attached block, or even the entire attached block, from the linear system, leading to a more memory-efficient procedure that also runs slightly faster in practice.
    \item Various parts of the pseudocode require access to the pivots in the active part of the coefficient block.
    Since they do not change until further activations, we can cache them when we start considering the next layer of vertices, skipping reads from part of the coefficient block.
    \item Many parts can be parallelised.
    For example, checking vectors in the attached block (or elsewhere if an earlier optimisation trick is implemented) for consistency, as well as finding corresponding solutions of the linear system, can all be parallelised.
    It is also possible to parallelise restoration of the echelon order, \ie to perform it once for all newly solved vertices simultaneously instead of doing it one by one.
    \item Lastly, when solving the linear systems in line~\ref{line:Ex=b} of the pseudocode, it is possible to only check for consistency (\ie reading whether relevant part of the column for considered attached vector is identically zero) without computing the actual solution. This can be particularly useful if one is interested only in a measurement labelling resulting in gflow (but not the witnessing gflow) or if one plans to look for focused gflow after the run.
\end{enumerate}

We note that these optimisations may also be applied to the existing $\bigO(n^3)$ procedure for flow finding where the measurement labelling is part of the input \cite{mitosekAlgebraicInterpretationPauli2026,uldemolins2026graphixsoftwareframeworkmeasurementbased}.

\end{document}